\documentclass[11pt,a4paper]{amsart}
\usepackage[foot]{amsaddr}
\usepackage{ifxetex}
\ifxetex
  \usepackage[no-math]{fontspec}
\else
\fi
\usepackage{amsmath}
\usepackage{amsfonts}
\usepackage{amssymb}
\usepackage{amsthm}
\usepackage{fullpage}
\usepackage{microtype}
\usepackage[libertine]{newtxmath}
\usepackage[tt=false]{libertine} 
\usepackage{caption}
\usepackage{bbm}
\usepackage{hyperref, color}
\hypersetup{colorlinks=true,citecolor=blue, linkcolor=blue, urlcolor=blue}
\usepackage[linesnumbered,boxed,ruled,vlined]{algorithm2e}
\usepackage{bm}
\usepackage{bbm}
\usepackage[numbers]{natbib}
\usepackage{xcolor}
\usepackage{enumerate} 
\usepackage{enumitem}
\usepackage{tabularx}
\usepackage{array}
\usepackage{cleveref}
\newcolumntype{L}[1]{>{\raggedright\arraybackslash}p{#1}}
\newcolumntype{C}[1]{>{\centering\arraybackslash}m{#1}}
\newcolumntype{R}[1]{>{\raggedleft\arraybackslash}p{#1}}

\usepackage{makecell}
\usepackage{footnote}
\usepackage{float}
\makesavenoteenv{tabular}

\renewcommand{\epsilon}{\varepsilon}

\newtheorem{theorem}{Theorem}[section]

\newtheorem*{claim*}{Claim}

\newtheorem{lemma}[theorem]{Lemma}
\newtheorem{proposition}[theorem]{Proposition}
\newtheorem{corollary}[theorem]{Corollary}
\theoremstyle{definition}

\newtheorem{definition}[theorem]{Definition}
\newtheorem{remark}[theorem]{Remark}
\newtheorem*{remark*}{Remark}

\def\Ex{\mathop{\mathbf{E}}\nolimits}

\renewcommand{\emptyset}{\varnothing}

\newcommand{\norm}[1]{\left\Vert#1\right\Vert}
\newcommand{\abs}[1]{\left\vert#1\right\vert}

\newcommand{\floor}[1]{\lfloor #1\rfloor}
\newcommand{\oq}{\overline{q}}
\newcommand{\BB}{{\bm B}}
\newcommand{\eps}{\varepsilon}

\newcommand{\NP}{\ensuremath{\mathbf{NP}}}

\newcommand{\RP}{\ensuremath{\mathbf{RP}}}
\newcommand{\defeq}{:=}

\newcommand{\wt}{\text{wt}}
\newcommand{\alphab}{\boldsymbol{\alpha}}
\newcommand{\betab}{\boldsymbol{\beta}}
\newcommand{\rb}{{\bm r}}
\newcommand{\cb}{{\bm c}}
\newcommand{\ib}{{\boldsymbol{i}}}
\newcommand{\emm}{\mathrm{e}}
\newcommand{\sgn}{\mathrm{sgn}}
\newcommand{\ones}{{\bm 1}}
\newcommand{\Jb}{{\bm J}}
\newcommand{\transpose}[1]{#1^{\texttt T}}

\def\*#1{\mathbf{#1}} 
\def\+#1{\mathcal{#1}} 
\def\-#1{\mathrm{#1}} 
\def\=#1{\mathbb{#1}} 

\def\MaxCut{\textnormal{\textsc{Max-Cut}}}
\def\MaxqCut{\textnormal{\textsc{Max-$q$-Cut}}}
\def\Mono{\operatorname{Mono}}
\def\Opt{\mathrm{Opt}}

\def\algbound{\frac{q^{K/3-1}}{4^KK^2}}
\def\searchbound{Kq^K}
\def\acbound{Kq^{K-1}}

\usepackage[textsize=tiny]{todonotes}
\usepackage{xifthen}

\makeatletter
\def\prob#1#2#3{\goodbreak\begin{list}{}{\labelwidth\z@ \itemindent-\leftmargin
                        \itemsep\z@  \topsep6\p@\@plus6\p@
                        \let\makelabel\descriptionlabel}
                \item[\it Name]#1
               \item[\it Instance]                #2
                \item[\it Output]#3
                \end{list}}
\makeatother

\newbool{doubleblind}
\setbool{doubleblind}{false}

\title{Inapproximability of counting hypergraph colourings}

\ifdoubleblind
  \author{Author(s)}
\else
\author{Andreas Galanis, Heng Guo, Jiaheng Wang}

\thanks{This project has received funding from the European Research Council (ERC) under the European Union's Horizon 2020 research and innovation programme (grant agreement No.~947778).}

\address[Andreas Galanis]{Department of Computer Science, University of Oxford, Wolfson Building, Parks Road, Oxford, OX1 3QD, United Kingdom. \textnormal{E-mail: \url{andreas.galanis@cs.ox.ac.uk}}}
\address[Heng Guo, Jiaheng Wang]{School of Informatics, University of Edinburgh, Informatics Forum, Edinburgh, EH8 9AB, United Kingdom. \textnormal{E-mail: \url{hguo@inf.ed.ac.uk}, \url{jiaheng.wang@ed.ac.uk}}}
\fi

\begin{document}

\begin{abstract}
Recent developments in approximate counting have made startling progress in developing fast algorithmic methods for  approximating the number of solutions to constraint satisfaction problems (CSPs) with large arities, using connections to the Lov\'asz Local Lemma. Nevertheless, the boundaries of these methods for CSPs with non-Boolean domain  are not  well-understood. Our goal in this paper is to fill in this gap and obtain strong inapproximability results by studying the prototypical problem in this class of CSPs,  hypergraph colourings. 

More precisely, we focus on the problem of approximately counting $q$-colourings on $K$-uniform hypergraphs with bounded degree $\Delta$. 
An efficient algorithm exists if $\Delta\lesssim \algbound$ (Jain, Pham, and Vuong, 2021; He, Sun, and Wu, 2021).
Somewhat surprisingly however, a hardness bound is not known even for the easier problem of finding colourings. For the counting problem, the situation is even less clear and there is no evidence of the right constant controlling the growth of the exponent in terms of $K$. 

To this end, we first establish that for general $q$ computational hardness for finding a colouring on simple/linear hypergraphs occurs at $\Delta\gtrsim \searchbound$, almost matching the algorithm from the Lov\'asz Local Lemma. Our second and main contribution is to obtain a far more refined bound  for the counting problem that goes well beyond the hardness of finding a colouring and which we conjecture is asymptotically tight  (up to constant factors). We show in particular that for all even $q\geq 4$ it is $\NP$-hard to approximate the number of colourings when $\Delta\gtrsim q^{K/2}$. Our approach is based on considering  an auxiliary weighted binary CSP model on graphs, which is obtained by ``halving'' the $K$-ary hypergraph constraints. This allows us to utilise reduction techniques available for the graph case, which hinge upon understanding the behaviour on random regular bipartite graphs that serve as gadgets in the reduction. The major challenge in our setting is to analyse the induced matrix norm of the interaction matrix of the new CSP which captures the most likely solutions of the system. In contrast to previous analyses in the literature, the auxiliary CSP  demonstrates both symmetry and asymmetry, making the analysis of the optimisation problem severely more complicated and demanding the combination of delicate perturbation arguments and careful asymptotic estimates.
\end{abstract}

\maketitle

\section{Introduction}

Constraint satisfaction problems (CSPs), such as $q$-colourings and  $k$-SAT, are perhaps the most-well studied problems in computer science.
We consider the case where the number of appearances of variables (also called degrees) is bounded by some absolute constant. For this class of CSPs, the Lov\'asz local lemma \cite{EL75} is a classical tool in combinatorics
that provides sharp degree thresholds under which the existence of solutions to CSPs is guaranteed.
After a long line of research \cite{Beck91, Alon91, MR98, CS00, Sri08, Mos09}, 
Moser and Tardos \cite{MT10} showed that, under the same conditions as the local lemma, an efficient algorithm exists to find a solution.
One remarkable aspect of this algorithm in the case of the bounded-degree $k$-SAT problem is that it gives up to lower order terms \cite{GST16} the location of the  algorithmic threshold for finding solutions \cite{onemore},  as the degree varies.

A related computational problem that has been intensively studied recently is to efficiently sample or approximately count solutions,\footnote{Approximate counting and sampling are often computationally equivalent, 
for example in the so-called ``self-reducible'' settings \cite{JVV86}. The local lemma setting is typically not self-reducible. However, reductions still exist \cite{FGYZ21,JPV20} between approximate counting and sampling without degradation of the parameters.} 
instead of merely finding one.
Under local lemma type conditions, there are some barriers, such as the connectivity barrier, for traditional approaches to approximate counting and sampling.
Recently, there have been some surprising developments that managed to bypass these barriers, making great progress towards a sampling version for the local lemma \cite{Moi19,GJL19,GLLZ19,FGYZ21,FHY20,JPV20,JPV21,he2021perfect}. However, there is no sharp threshold established in the sampling setting yet.

A major difference between searching and sampling is that sampling can be computationally harder in the local lemma settings.
For example, for $K$-CNF formulas where each variable appears at most $\Delta$ times,
if $\Delta\le \frac{2^K}{eK}+1$, then there must be a satisfying assignment, and it can be efficiently found;
yet if $\Delta\ge 5\cdot2^{K/2}$, there is no algorithm to sample or approximately count satisfying assignments unless $\NP=\RP$ \cite{BGGGS19},
even when all variables appear positively (the monotone case). 
For monotone $K$-CNFs, the threshold $2^{K/2}$ is sharp up to constants, 
because Hermon, Sly, and Zhang \cite{HSZ19} showed a complementary efficient algorithm for $\Delta\le c2^{K/2}$ where $c$ is a constant.

Our goal in this paper is to detail how this ``sampling-is-computationally-harder'' phenomenon manifests into local-lemma-type hypergraph problems with non-boolean domain and which are not necessarily monotone,  
and to make progress towards carving the computational thresholds for sampling problems in the local lemma setting.
Among the more recent algorithmic developments,
the most promising one to establish a computational transition is the problem of counting hypergraph colourings,
which was the original setting where the local lemma was developed \cite{EL75}.
For a hypergraph $H=(V,E)$, a proper $q$-colouring $\sigma:V\rightarrow \{1,2,\dots,q\}$ assigns a colour to each vertex,
such that no hyperedge is monochromatic under $\sigma$.
Suppose further that $H$ is $K$-uniform and $\Delta$ is the maximum degree of $H$.
In this setting, an efficient sampling algorithm exists if $\Delta\lesssim  \algbound$\footnote{Note that in \cite{JPV21,he2021perfect}, their $\Delta$ is the degree bound for the dependency graph, which is at most $\Delta(K-1)$ in our setting.} \cite{JPV21,he2021perfect},
where $\lesssim$ (and similarly, $\gtrsim$, $\asymp$) hides logarithm or other lower order terms.

Somewhat surprisingly, despite being a canonical and well-studied  problem,
not much is known regarding the computational hardness for bounded-degree hypergraph colourings, even for the search version.
Thus we first show that it is $\NP$-hard to find a proper hypergraph colouring if $q\ge 2$, $K\ge 2$ (but not $q=K=2$), and $\Delta \gtrsim \searchbound$ (see \Cref{thm:search-hardness}), 
and to approximately count if $q\ge 2$, $K\ge 4$, and $\Delta\gtrsim \acbound$ (see \Cref{thm:counting-simple-hardness}).
These bounds almost match the algorithmic Local Lemma threshold.
In fact, \Cref{thm:search-hardness} and \Cref{thm:counting-simple-hardness} still hold when restricting to simple\footnote{A hypergraph is called \emph{simple} (or linear) if the intersection of any two hyperedges has size at most $1$.} hypergraphs.
If we restrict monotone $K$-CNFs to simple hypergraphs, 
the condition of the aforementioned algorithmic result \cite{HSZ19} improves to $\Delta\lesssim \frac{2^K}{K^2}$.
In view of this result and the searching algorithm~\cite{MT10}, 
it seems reasonable to conjecture that the sharp hardness threshold (for both approximate counting and searching) in simple hypergraphs is $\Delta \gtrsim q^{K-1}$,
up to some polynomial factors in $K$.
Our hardness result almost matches it.

Our second and main contribution is a more refined hardness result for approximate counting and sampling, stated as follows.
\begin{theorem}  \label{thm:main-hardness}
  Let $q\ge 4$ be even, $K\ge 4$ be even, and $\Delta\geq 5 q^{K/2}$. It is $\NP$-hard to approximate the number of proper $q$-colourings in $n$-vertex $K$-uniform hypergraphs of maximum degree at most $\Delta$, even within a factor of $2^{cn}$ for some constant $c(q,K)>0$.
\end{theorem}
A few remarks are in order.
\begin{itemize}
  \item First, the threshold in \Cref{thm:main-hardness} is far more refined than the corresponding theorem (\Cref{thm:search-hardness}) for the searching problem.
    The exponent of the sampling threshold we achieve is roughly half of that of the threshold for the searching problem,
    which is analogous to the aforementioned (monotone) $K$-CNF example \cite{BGGGS19}.
    Interestingly, and in contrast to the $K$-CNF case, our colourings threshold is getting close to matching the constant in the exponent in the algorithmic threshold of $\Delta\lesssim  \algbound$ \cite{JPV21,he2021perfect}.
    We conjecture that our hardness threshold is asymptotically tight (up to constant factors),
    namely that for all $q\ge 2$ and $K\ge 2$, there is an efficient algorithm to approximately count the number of proper $q$-colourings in $K$-uniform hypergraphs whenever $\Delta\lesssim q^{K/2}$.
  \item Second, our result applies to only even $K$ for $K$-uniform hypergraphs.
    This is due to a particular halving construction we use in the reduction.
    The hardness results for (monotone) $K$-CNF~\cite{BGGGS19} allow hyperedges with sizes at least $K$.
    This is a stronger assumption and our hardness bound would still apply without changing the order.
    In fact, we expect a slight variant of our construction to work for odd $K$ to achieve a threshold of the same order. (See \Cref{rmk:odd-K}.)
    As we explain soon, the details for even $K$ are already very complicated,
    so for clarity and simplicity we did not pursue the odd $K$ case.
  \item Lastly, our result applies to an even number of colours $q$,     which is analogous to  hardness results for counting in the graph colouring setting \cite{galanis2015inapproximability}. 
    It was left as an open problem in \cite{galanis2015inapproximability} to handle odd $q$ (see also the recent work \cite{polymerschen}),
    and we met the same difficulty in our setting as well. Our hardness proof for counting builds on ideas from \cite{galanis2015inapproximability}, and we focus on the challenges needed to refine them in the hypergraph setting (rather than addressing the parity of $q$). 
    We expect that substantial new ideas are required to resolve the parity of $q$, even in the graph setting.
\end{itemize}

In order to show \Cref{thm:main-hardness},
we first reduce from an auxiliary weighted binary CSP, 
namely a ``spin system'' in graphs. 
Basically, we replace each vertex of the graph by a cluster of $K/2$ vertices in the hypergraph,
and an edge by a hyperedge of size $K$.
This construction is identical to the one in \cite{BGGGS19},
via which one reduces from weighted independent set in graphs to hypergraph independent sets.
However, in order to reduce to the hypergraph $q$-colouring problem, 
the variables of the weighted binary CSP take $q+1$ possible values.
There are $q$ values that correspond to ``pure'' colours,
and one special value that corresponds to a ``mixed'' state.
The interactions among these $q+1$ states are dictated by the hypergraph colouring problem,
and the mixed state behaves very differently from the pure colours; roughly, the pure colours behave symmetrically (as in the graph case) but the mixed state causes asymmetry.

Our next and main step is to show the desired hardness result for this spin system.
We follow an established route of establishing inapproximability for spin systems \cite{DFJ02,mossel2009hardness,Sly10,CCGL12,SS14,galanis2016inapproximability},
and in particular \cite{galanis2015inapproximability},
where the key is to understand the system on random regular bipartite graphs which are used as gadgets in the reduction.
More precisely, we need to analyze what are the most likely configurations of the system  on random regular bipartite graphs,   the so-called dominant phases (given by the normalised counts of the colours on each side of the graph). It was shown in \cite{galanis2015inapproximability} that these are captured by a certain matrix norm of the interaction matrix. These norms are in general very hard to penetrate analytically and it was already a major difficulty in the perfectly symmetric setting of  \cite{galanis2015inapproximability}.
For us, the presence of a special spin together with $q$ symmetric spins makes our spin system very different from all of the spin systems analyzed before and the mixture of symmetry and asymmetry make the analysis substantially harder.
For example, in \cite{galanis2015inapproximability}, to show that the two parts of the graph are unbalanced,
a simple Hessian calculation suffices,
whereas in our setting, there are Hessian stable balanced phases due to the presence of this special spin (that can be favoured against the others).
Also, being Hessian stable means that this phase is locally maximal,
making perturbation arguments hard to carry out.
What we do instead is to directly compare this phase with the dominant phase via a careful interpolation path and a sequence of delicate estimates. This reflects the main difference between our work and previous works,
namely that our estimates and perturbation arguments are significantly more delicate in order to rule out the local-maxima.

The main open problem left is to close the gap between the algorithm of \cite{JPV21,he2021perfect} and our hardness threshold, \Cref{thm:main-hardness},
although we expect that any progress towards a computational transition threshold now would come from the algorithmic side.
Another open problem is to handle the odd $q$ case in \Cref{thm:main-hardness} and similarly in \cite{galanis2015inapproximability}.
When $q$ is odd, the current perturbation-based analysis cannot determine, among a few candidates, which phases are dominant.
New ideas would be required to handle the odd $q$ case.

\section{Hardness for finding colourings in simple hypergraphs}

In this section we show hardness results for finding hypergraph colourings for parameters beyond the local lemma condition.
The key is to find instances that do not have proper colourings.

We will use a configuration model to construct random regular hypergraphs.
With constant probability, the resulting hypergraph is simple \cite{randomhyper,panagiotou}.
Frieze and Mubayi \cite{FM13} showed that if $q>c \left(\frac{\Delta}{\log \Delta}\right)^{\frac{1}{K-1}}$ for some constant $c=c(K)$ that only depends on $K$, 
then any simple $K$-uniform hypergraph with maximum degree $\Delta$ is $q$-colourable.
In particular, their condition holds if $\Delta\le c Kq^{K-1}\ln q$ for some constant $c=c(K)$.
Our next lemma complements their result by showing as an intermediate result that if $\Delta>K q^{K-1}\ln q+1$,
we can find a $K$-uniform hypergraph with maximum degree $\Delta$ which is not $q$-colourable. For our reductions, we use such hypergraphs to obtain a ``disequality'' gadget, as detailed in the lemma below.
\begin{lemma}\label{lem:notsatisfiable}
  Let $q,K\geq 2$ be integers. 
  Then, for all integers $\Delta> K q^{K-1}\ln q+1$, there exists a $q$-colourable $K$-uniform simple hypergraph $H$ with maximum degree $\Delta$ and two distinct vertices $u,v$
  such that the degree of $u$ is $1$, the degree of $v$ is at most $\Delta$,
  and for every $q$-colouring $\sigma$ of $H$ it holds that $\sigma(u)\neq \sigma(v)$.
\end{lemma}
\begin{proof} 
  We first argue that for all $\Delta> K q^{K-1}\ln q$  there is a $\Delta$-regular hypergraph $H_0$ such that $Z_{col}(H_0)=0$,
  where $Z_{col}(H)$ denotes the number of $q$-colourings in $H$.

Let $n$ be such that $m=n\Delta/K$ is an integer. We sample a $K$-uniform $\Delta$-regular hypergraph $H$ according to the following pairing model (see \cite{panagiotou}). Start with a bipartite graph with the points $[n]\times [\Delta]$ on the left and the points $[m]\times [K]$ on the right, and pair the two sides using a uniformly random perfect matching; the vertex set of the  final hypergraph $H$  is obtained in the natural way by projecting the set $[n]\times [\Delta]$ onto $[n]$. Note that it will be convenient to view the hyperedges of $H$ for now as ordered tuples rather than sets; this does not make any difference when considering colourings of $H$ due to the symmetry among possible ordering of the colours within the hyperedge. It is a well-known fact, see for example  \cite[Lemma 2]{randomhyper} or \cite[Theorem 2.4 \& Appendix A.4]{panagiotou}, that the probability that $H$ is simple is bounded away from zero for all sufficiently large $n$.\footnote{We remark that the term ``simple'' has different meanings across the literature. A simple hypergraph in this paper actually corresponds to a configuration without $4$-cycles in the context of \cite[Theorem 2.4]{panagiotou} (where one should plug in $\ell=2$), or a hypergraph without $2$-cycles in \cite{randomhyper}.}

For a colouring $\sigma:[n]\rightarrow[q]$, a colour $i\in [q]$ and a $K$-tuple of colours $\ib=(i_1,\hdots,i_K)\in[q]^K$, let $n\alpha_i$ be the number of vertices with colour $i$, and $m\beta_{\ib}$ be the number of hyperedges whose vertices are coloured according to $\ib$ (i.e., the $j$-th vertex of the hyperedge takes the colour $i_j$). Let $\alphab=\{\alpha_i\}_{i\in [q]}$ and $\betab=\{\beta_{\ib}\}_{\ib\in [q]^K}$, and note that $(\alphab,\betab)\in S_q$, where $S_q$ is the space of all pairs of vectors in $\mathbb{R}^q\times \mathbb{R}^{q^K}$ satisfying 
\begin{equation}\label{eq:Sq}
\begin{gathered}
\mbox{$\sum_{i\in [q]}$}\alpha_i=1,\quad \mbox{$\sum_{\ib\in [q]^K}$}t_{i,\ib}\beta_{\ib}=K\alpha_i \mbox{ for } i\in [q]\\
\alpha_i\geq 0 \mbox{ for } i\in [q],\quad \beta_{\ib}\geq 0 \mbox{ for } \ib\in [q]^K,\quad \beta_{(i,i,\hdots,i)}=0 \mbox{ for } i\in [q],
\end{gathered}
\end{equation}
where for $i\in [q]$ and $\ib\in [q]^K$ we denote by $t_{i,\ib}$ the number of occurrences of colour $i$ in the tuple $\ib$. 
Then, we have 
\[\mathbf{E}[Z_{col}(H)]=\frac{1}{(\Delta n)!}\sum_{(\alphab,\betab)\in S_q;\, n\alphab\in \mathbb{Z}^q, m\betab\in \mathbb{Z}^{q^K}}  \binom{n}{\alpha_1n,\hdots,\alpha_qn}\binom{m}{\beta_1m,\hdots,\beta_{q^K}m}\prod_{i\in [q]} (\Delta \alpha_i n)!,\]
since a term in the sum corresponding to $(\alphab,\betab)$ accounts for the number of ways to choose $\sigma$ and $H$ with  vertex-colour frequencies given by the vector $\alphab$ and edge-colour frequencies given by the vector $\betab$.  Using Stirling's approximation $(2\pi k)^{1/2}(k/\emm)^k\leq k!\leq \emm k^{1/2} (k/\emm)^{k}$ that holds for all integers $k\geq 1$, we obtain by  expanding the terms inside the sum (note that there are at most $n^{q^K+q}$ of them) that
\begin{equation}\label{eq:Zupperbound}
\mathbf{E}[Z_{col}(H)]\leq n^{O(1)} \exp\Big(n \max_{(\alphab,\betab)\in S_q}F(\alphab,\betab)\Big),
\end{equation}
where $F(\alphab,\betab)=-(\Delta-1)h(\alphab)+\frac{\Delta}{K}h(\betab)$ and $h(\cdot)$ is the entropy function (here, we adopt the usual convention that $0\ln 0=0$ which makes $h$ and $F$ continuous and therefore the maximum in \eqref{eq:Zupperbound} well-defined).

For $(\alphab,\betab)\in S_q$, we have that $\alpha_i=\tfrac{1}{K}\mbox{$\sum_{\ib\in [q]^K}$}t_{i,\ib}\beta_{\ib}$ for $i\in [q]$, and hence
\[F(\alphab,\betab)=h(\alphab)+\tfrac{\Delta}{K} G(\alphab,\betab) \mbox{ where } G(\alphab,\betab)=h(\betab)-\mbox{$\sum_{i\in [q]}$}\ln(\alpha_i)\mbox{$\sum_{\ib\in [q]^K}$} t_{i,\ib}\beta_\ib. \]
Note that for a fixed vector $\alphab$, the function $G_{\alphab}(\betab):=G(\alphab,\betab)$ is concave and the method of  Lagrange multipliers yields that the maximum of $G_{\alphab}$  happens at $\betab^*=\{\beta_\ib^*\}_{\ib\in[q]^K}$ that satisfies
\[\beta_{\ib}^*=\frac{\prod_{i\in [q]}(\alpha_i)^{t_{i,\ib}}\prod_{i\in [q]}\mathbf{1}_{\ib^*\neq (i,i,\hdots,i)}}{1-\norm{{\alphab}}^K_K} \mbox{ for $\ib\in [q]^K$}, \quad G_{\alphab}({\betab}^*)=\ln(1-\norm{\alphab}^K_K).\]
It follows that
\begin{equation}\label{eq:Fab243}
F(\alphab,\betab)\leq h(\alphab)+\tfrac{\Delta}{k}\ln(1-\norm{\alphab}^K_K)\leq \ln\Big( q\big(1-\tfrac{1}{q^{K-1}}\big)^{\Delta/K}\Big),
\end{equation}
where the last inequality follows from $h(\alphab)\leq \ln q$ and $\norm{\alphab}^K_K\geq 1/q^{K-1}$, both of which are simple applications of Jensen's inequality. For $\Delta> K q^{K-1}\ln q$, the r.h.s.~of \eqref{eq:Fab243} is negative and therefore $\max_{(\alphab,\betab)\in S_q}F(\alphab,\betab)<0$. From \eqref{eq:Zupperbound}, we conclude that $Z_{col}(H)=0$ with probability $1-\exp(-\Omega(n))$. By a union bound, we obtain a simple $\Delta$-regular hypergraph $H_0$ with $Z_{col}(H_0)=0$, as claimed.

To obtain the final hypergraph $H$ with the desired property, we invoke an argument in \cite[Lemma 28]{hyper2spin} (which in turn was inspired by \cite{onemore}).
We give the details here for completeness.
Given $H_0=(V,\+E)$ with $Z_{col}(H_0)=0$,
we can remove hyperedges from $\+E$ one by one until removing any more hyperedge makes $Z_{col}(H)>0$.
Call the resulting hypergraph $H_0'=(V,\+E_0')$.
Clearly $H_0'$ is simple and has at least one hyperedge.

Choose an arbitrary hyperedge $e\in \+E_0'$.
Let $S\subseteq e$ be the set of vertices with non-zero degree in $H_0'-e$.
If $S = \emptyset$, then $e$ is disconnected from the rest of the graph.
Thus as $H_0'$ is not $q$-colourable, removing $e$ would not make the hypergraph $q$-colourable.
This contradicts to the minimality of $H_0'$ and thus $S\neq \emptyset$.
Denote the vertices in $S$ by $v_1,\dots,v_i$, and the vertices in $e\setminus S$ by $v_{i+1},\dots,v_K$.
We construct $i$ simple hypergraphs $H_1,\dots,H_i$ where for $1\le j\le i$, 
in $H_j$ we introduce new vertices $u_1,\dots,u_j$ and replace the hyperedge $e$ by $e_j\defeq\{u_1,\dots,u_j,v_{j+1},\dots,v_K\}$.
By minimality of $H_0'$ again, $Z_{col}(H_i)>0$ as $e_i$ is disconnected from the rest of $H_i$.
Thus we can find the smallest $j\ge 1$ such that $Z_{col}(H_j)>0$ and $Z_{col}(H_{j-1})=0$ (or $Z_{col}(H_{0}')=0$ if $j=1$).
For any proper colouring $\sigma$ of $H_j$, 
if $\sigma(u_j)=\sigma(v_j)$,
$\sigma$ would be a proper colouring of $H_{j-1}$,
contradicting to the above.
Thus it must hold that for any colouring $\sigma$ of $H_j$, $\sigma(u_j)\neq\sigma(v_j)$.
This is the hypergraph required by the lemma,
with $u=u_j$ and $v=v_j$.
Moreover, the degree of $u_j$ is $1$, and the degree of $v_j$ is at most $\Delta$.
\end{proof}

\Cref{lem:notsatisfiable} leads to the following hardness result,
where we lose a factor $2q$ in the degree bound due to the reduction.
We note that for $q=2$, $K=3$, $\Delta=4$, and simple hypergraphs, $\NP$-hardness is known \cite{DD20}.
However, the main point of the next theorem is that there is a degree bound that scales roughly as $q^K$ and makes the problem $\NP$-hard.

\begin{theorem}\label{thm:search-hardness}
Let $q,K\geq 2$ be integers with $(q,K)\neq (2,2)$. Then, it is $\NP$-hard to find a $q$-colouring on a $K$-uniform simple hypergraph of maximum degree at most $\Delta$, when $\Delta\geq 2Kq^{K}\ln q+2q$.
\end{theorem}
\begin{proof}
For $q>2$, we reduce from the problem of finding $q$-colourings in graphs whose degrees are bounded by $2q$.
The latter problem is shown to be $\NP$-hard by \cite{EHK98}.
Given a graph $G$, we replace each edge $(u,v)$ of $G$ by the hypergraph in \Cref{lem:notsatisfiable},
where $u$ and $v$ are identified with the special vertices in the hypergraph.
Then each such hypergraph is effectively a disequality for the colours of $u$ and $v$.
Call the resulting hypergraph $H$.
Thus $G$ is $q$-colourable if and only if $H$ is $q$-colourable.
The maximum degree of $H$ is $2q(Kq^{K-1}\ln q+1)$.

For $q=2$ and $K>2$, using two copies of the hypergraph from \Cref{lem:notsatisfiable}, we build an ``equality'' gadget, i.e., a simple hypergraph $H$ of maximum degree $\Delta\leq 2(Kq^{K-1}\ln q+1)=K 2^K\ln2+2$ with distinct vertices $u,v$  which both have degree 1 such that for every $q$-colouring $\sigma$ it holds that $\sigma(u)=\sigma(v)$. It is well-known that finding 2-colourings of $K$-uniform simple hypergraphs is $\NP$-hard (or we can use for example \cite{DD20}), and using the equality gadget $H$, for any $K$-uniform simple hypergraph $F$, we can construct a $K$-uniform simple hypergraph $F'$ of maximum degree $\Delta$ such that $F$ is 2-colourable if and only if $F'$ is 2-colourable. 
One possible way to do so is replacing each degree-$d$ vertex $w$ of $F$ with a cycle of length $d$ and then replacing each edge $e$ of the cycle with a distinct copy of the hypergraph $H$ using $u,v$ for the endpoints of the edge $e$; then, for each hyperedge of $F$ that uses $w$, in $F'$ we use instead one of the $d$ vertices of the cycle.
\end{proof}

Note that the result of Frieze and Mubayi \cite{FM13} is also algorithmic.
Thus \Cref{thm:search-hardness} is sharp for simple hypergraphs up to a factor $c q$ where $c=c(K)$ is a constant depending only on $K$.
For general hypergraphs, the algorithm of Moser and Tardos \cite{MT10} applies in this setting when $\Delta\le \frac{q^{K-1}}{e(K-1)}$,
in which case \Cref{thm:search-hardness} almost matches the algorithmic result, up to a factor of $cK^2 q \ln q$ where $c$ is a constant.

For approximate counting, we can avoid the loss of the factor $q$ when $q\ge 2$ and $K\ge 4$.
For a $q$-by-$q$  matrix $\BB=\{B_{ij}\}_{i,j\in [q]}$, the partition function for the $q$-spin system with interaction matrix $B$ in a graph $G=(V,E)$ is given by
\begin{align}\label{eqn:Z_B}
  Z_B(G)\defeq\sum_{\sigma:V\rightarrow \{1,\dots,q\}}\wt(\sigma),
\end{align}
where $\wt(\sigma)\defeq\prod_{(u,v)\in E}B_{\sigma(u)\sigma(v)}$ is the weight of an assignment $\sigma:V\rightarrow \{1,\dots,q\}$ of the $q$ spins to the vertices of $G$. In particular, the $q$-state antiferromagnetic Potts model corresponds to the case where $\BB$ is the matrix whose off-diagonal entries are equal to 1, whereas the diagonal entries equal to some parameter $B<1$ (note, $B=0$ corresponds to $q$-colourings).

We will use the following hardness result about the Potts model.
A \emph{fully polynomial-time randomized approximation scheme} (FPRAS) is an algorithm that takes the accuracy $\eps$ as an extra input,
outputs an $\eps$-approximation, and runs in time polynomially bounded by both the instance size and $1/\eps$.
\begin{lemma}  \label{lem:Potts}
  There is a constant $C_1>5$ such that,
  for any integers $q\ge 2$, $\Delta\ge 2C_1q\ln q$, and $B<1-\frac{C_1 q\ln q}{\Delta}$,
  there is no FPRAS to approximate the $q$-state antiferromagnetic Potts partition function $Z_B$ in graphs with bounded degree $\Delta$,
  unless $\NP=\RP$.
\end{lemma}

The proof of \Cref{lem:Potts} is quite a detour from the problems we focus on, 
so we postpone it to \Cref{sec:Potts}.
We note that \Cref{lem:Potts} is weaker than the inapproximability result in \cite[Theorem 1.2]{galanis2015inapproximability},
which achieves $B<1-\frac{q}{\Delta}$ but only holds for even $q$.
We want to deal with general $q$, and thus settle with this weaker version.

\begin{theorem}\label{thm:counting-simple-hardness}
  There is a constant $C_1>5$ such that,
  for any integers $q\ge 2$, $K\ge 4$, and $\Delta\ge C_1Kq^{K-1}\ln q$,
  unless $\NP=\RP$, there is no FPRAS for the number of 
  $q$-colourings in $K$-uniform simple hypergraphs of maximum degree at most $\Delta$.
\end{theorem}
\begin{proof}
  We reduce the partition function of the $q$-state antiferromagnetic Potts model with $B=1-\frac{1}{q^2-3q+3}$ in graphs with bounded degree $\Delta$
  to the problem of counting $q$-colourings in $K$-uniform simple hypergraphs of maximum degree at most $\Delta$.
  Note that if $K\ge 4$ and $\Delta\geq C_1Kq^{K-1}\ln q$, where 
  $C_1$ is from \Cref{lem:Potts},
  then $B<1-\frac{C_1 q\ln q}{\Delta}$.
  Thus the reduction implies the theorem via \Cref{lem:Potts}.

  The reduction goes as follows.
  For each edge $(u,v)$ in a $\Delta$-regular graph $G=(V,E)$,
  we replace it by a gadget using the hypergraph $H$ in \Cref{lem:notsatisfiable},
  whose degree bound is $\Delta_0=Kq^{K-1}\ln q+1$.
  To be more specific, we introduce new vertices $w_1$ and $w_2$.
  We add three copies of the hypergraph $H$ with special vertices $(u,w_1)$, $(w_2,w_1)$, and $(v,w_2)$, respectively.
  Do this for all edges in $G$.
  Then, the degrees of $u$ and $v$ are still $\Delta$,
  the degrees of $w_1$'s are at most $2\Delta_0 < \Delta$,
  and the degrees of $w_2$'s are at most $\Delta_0+1 < \Delta$.
  All other newly introduced vertices have degrees at most $\Delta_0 < \Delta$.
  Thus, the degree requirement is met.
  Call the resulting hypergraph $H_G$.

  To finish the reduction,
  we claim that
  \begin{align*}
    Z_{col}(H_G)=C^{|E|}Z_B(G),
  \end{align*}
  where $C$ is a constant depending only on $H$.
  First notice that for any pair of colours $i$ and $j$,
  the number of colourings $\sigma$ of $H$ such that $\sigma(u)=i$ and $\sigma(v)=j$ is a constant,
  due to the symmetry among colours.
  Denote this constant by $C_0$.
  Thus, in the gadget above,
  when the two endpoints $u$ and $v$ have different colours, the number of possible colourings for the gadget is $((q-2)^2+(q-1))C_0^3$;
  when the two endpoints $u$ and $v$ have the same colour, the number of possible colourings for the gadget is $(q-1)(q-2)C_0^3$.
  The claim holds with $C=((q-2)^2+(q-1))C_0^3$.
\end{proof}

In \Cref{thm:counting-simple-hardness}, we could avoid the large constant $C_1$ in the degree bound by using \cite[Theorem 1.2]{galanis2015inapproximability} as the starting point of our reduction,
but doing so will restrict the result to even $q$ only.

For $K$-CNFs on simple hypergraphs,
Hermon, Sly, and Zhang \cite{HSZ19} showed an efficient approximate counting and sampling algorithm if $\Delta\le \frac{c2^K}{K^2}$, where $c$ is a constant.
In view of their result, \Cref{thm:search-hardness} and \Cref{thm:counting-simple-hardness} are potentially sharp for simple hypergraphs, up to some polynomial factor in $K$.

\section{Refined inapproximability for approximate counting}

In this section we show our main theorem, \Cref{thm:main-hardness}, namely a refined inapproximability result for counting. As mentioned earlier, we will do this by first relating it to a multi-spin system on  graphs with ``antiferromagnetic'' interaction matrix $\BB$, and then establishing inapproximability results. 
It is tempting to pursue a strategy similar to that of \Cref{lem:Potts} to show hardness for the spin system defined by $\BB$.
However, that strategy relies on hardness of finding the maximum weight configuration, and somewhat surprisingly, as we shall see soon, that problem for $\BB$ is trivial. 
Instead, we need sharper tools from \cite{galanis2015inapproximability}.

To define the spin system on graphs we will be interested in, we only need to specify its interaction matrix $\BB$ (recall \eqref{eqn:Z_B}).  We use $[q]$ to denote $\{1,\dots,q\}$ and $[\overline{q}]$ to denote $\{0,1,\dots,q\}$ . Let $t\defeq (q^k-q)^{1/\Delta}$, where $k\defeq K/2$, and ${\bm B}=\{B_{ij}\}_{i,j\in[\oq]}$ be the matrix with block form
\begin{align}\label{eq:BBdef}
  \BB=\bigg[\begin{array}{cc} t^2 &  t\transpose{\ones} \\ t\ones &  \Jb \end{array}\bigg],
\end{align}
where $\Jb$ is the $q\times q$ matrix with 0s on the diagonal and 1s elsewhere, and $\ones$ is the $q\times 1$ vector with all ones. In the language of \cite{galanis2015inapproximability}, the matrix $\BB$ is antiferromagnetic and ergodic.\footnote{Antiferromagnetism amounts to checking that $\BB$ has all but one of its eigenvalues negative; it is not hard to see that $\BB$ has $-1$ as an eigenvalue by multiplicity $q-1$, and therefore using trace/determinant we see that the other two eigenvalues have sum equal to $q-1+t^2$ and product $-t^2$. Ergodicity amounts to the fact that $\BB$ is irreducible and aperiodic.}

Let $H$ be a $K$-uniform hypergraph, where $K=2k$ is even, and recall that we use $Z_{col}(H)$ to denote the number of proper $q$-colourings of $H$. 
For any given $\Delta$-regular graph $G=(V,E)$, let $H_G$ be the hypergraph where every vertex $v\in V$ is replaced by $k$ new vertices $v_1,\dots,v_k$,
and each edge $(u,v)$ is replaced by a hyperedge $\{u_1,\dots,u_k,v_1,\dots,v_k\}$ of size $2k$.
Then $H_G$ is $2k$-uniform and $\Delta$-regular.
This construction has been used in~\cite{BGGGS19}, and yields the following lemma in our case.
\begin{lemma}  \label{lem:graph-to-hypergraph}
  Let $G=(V,E)$ be a $\Delta$-regular graph, and $H_G=(V',E')$ be the $2k$-uniform hypergraph constructed as above.
  Then, $    Z_{\bm B}(G) = Z_{col}(H_G)$.
\end{lemma}
\begin{proof}
  Let $\Omega_\BB$ be the set of all assignments $\sigma$ of $G$ whose weights are non-zero.
  Let $\Omega_{col}$ be the set of all proper $q$-colourings $\tau$ of $H$.
  We will construct a surjective mapping $\varphi$ between $\Omega_{col}$ and $\Omega_B$,
  such that for any $\sigma$, $\abs{\varphi^{-1}(\sigma)} = \wt(\sigma)$.
  This implies the lemma.

  The mapping $\varphi$ is as follows.
  Given $\tau:V'\rightarrow\{1,2,\dots,q\}$,
  let
  \begin{align*}
    \varphi(\tau)(v)\defeq
    \begin{cases}
      i & \text{if $\tau(v_1)=\tau(v_2)=\dots=\tau(v_k)=i$ for some $1\le i\le q$,}\\
      0 & \text{otherwise.}
    \end{cases}
  \end{align*}
  We first show that $\varphi$ is surjective.
  Let $\sigma\in\Omega_B$ and we construct $\tau\in \Omega_{col}$ such that $\varphi(\tau)=\sigma$.
  For any $v$ such that $\sigma(v)\neq 0$,
  $\tau(v_i)=\sigma(v)$  for any $1\le i\le k$.
  If $\sigma(v)=0$,
  then let $\tau(v_i)=1$ for any $1\le i\le k-1$, and $\tau(v_k)=2$.
  It is easy to verify that $\tau$ is a proper $q$-colouring and $\varphi(\tau)=\sigma$ for this construction.

  Next we calculate $\abs{\varphi^{-1}(\sigma)}$.
  Let $n_0(\sigma)$ be the number of vertices assigned $0$ under $\sigma$.
  Then
  \begin{align*}
    \abs{\varphi^{-1}(\sigma)} = \left(q^k-q\right)^{n_0(\sigma)}.
  \end{align*}
  On the other hand,
  since $G$ is $\Delta$-regular,
  \begin{align*}
    \wt(\sigma) = t^{\Delta n_0(\sigma)} = \left(q^k-q\right)^{n_0(\sigma)} = \abs{\varphi^{-1}(\sigma)},
  \end{align*}
  which verifies the properties of $\varphi$.
\end{proof}
\begin{remark}\label{rmk:odd-K}
  For odd $K$,
  we may consider a similar construction, but in addition to clustering half of each hyperedge as a single vertex,
  we leave one vertex in the middle which appears only in this single hyperedge.
  The resulting spin system would have a different matrix $\BB'$,
  but the difference between $\BB'$ and the current $\BB$ is not too much in the sense that the zeros would be replaced by small constants.
  We expect that we may obtain a hardness result for $\BB'$ for $\Delta$ of a similar order.
  However, since the details are already getting very complicated,
  we will only handle $\BB$ in the rest of this paper.
\end{remark}

Given \Cref{lem:graph-to-hypergraph},
all we need to show is that the spin system with interaction matrix $\BB$ is hard to approximate on $\Delta$-regular graphs,
with $\Delta$ in the desired range.
For this, we will use a result by Galanis, Vigoda, and {\v{S}}tefankovi{\v{c}} \cite[Theorem 1.5]{galanis2015inapproximability} which gives a sufficient condition in terms of studying a certain function (that can be formulated in terms of an induced norm of $\BB$).
Note that since $t>1$ the corresponding optimization problem related to $\BB$ is trivial. 
Thus, we cannot use a strategy similar to that of \Cref{lem:Potts} to show hardness for the spin system defined by $\BB$.

The main construction in the gadget to show the hardness is the bipartite random $\Delta$-regular graph.
Let $(\alphab,\betab)$ be a pair of vectors such that for $i\in [\oq]$, $\alpha_i$ and $\beta_i$ denotes the fraction of vertices with colour $i$ on the left and right sides of the bipartite random regular graph.
If we draw a sample $\sigma$ proportional to its weight $\wt(\sigma)$,
then with high probability over the choice of the random graph,
the fraction of colours $(\alphab,\betab)$ will be from one of the dominant phases, for all but an exponentially small probability.
Analyzing these dominant phases lies in the heart of \cite[Theorem 1.5]{galanis2015inapproximability}.

Let $\+G_n$ denote the family of $\Delta$-regular bipartite graphs with $n$ vertices on each side. For a bipartite graph $G$ uniformly drawn from $\+G_n$ and probability vectors $\alphab=\{\alpha_i\}_{i\in [q]}, \betab=\{\beta_i\}_{i\in [q]}$,   we use $Z_\BB^{\alphab,\betab}(G)$ to denote the total weights of assignments whose fractions of colours on the two sides are given by $\alphab,\betab$ respectively.
Consider the function $\Psi_1$ that captures the exponential growth of the expectation of $Z_\BB^{\alphab,\betab}(G)$, i.e.,
\begin{equation}\label{eq:Psipsi}
\Psi_1(\alphab,\betab)\defeq\lim_{n\rightarrow\infty}\frac{1}{n}\log\Ex_{\+G_n}[Z_{\BB}^{\alphab,\betab}(G)].
\end{equation}
The function $\Psi_1$ has a relatively explicit form (see \cite[Section 2]{galanis2015inapproximability}) using entropy-style functions though the exact details are not going to be important and we will in fact use a surrogate function later on (see Section~\ref{sec:dominant-phase}).

Before stating the main result of \cite{galanis2015inapproximability}, we need some further terminology. A \emph{dominant phase} $(\alphab,\betab)$ is a maximizer of the function $\Psi_1(\alphab,\betab)$ and captures the most likely configurations for the spin system with interaction matrix on a random $\Delta$-regular graph. A dominant phase is called \emph{Hessian dominant} if the Hessian of $\Psi_1$ is negative definite. Finally, two dominant phases $(\alphab_1,\betab_1)$ and $(\alphab_2,\betab_2)$ are \emph{permutation symmetric} if 
there is a permutation matrix ${\bm P}$ such that $\BB={\bm P}\BB\transpose{{\bm P}}$ and $(\alphab_1,\betab_1)=({\bm P}\alphab_2,{\bm P}\betab_2)$ or $(\alphab_1,\betab_1)=({\bm P}\betab_2,{\bm P}\alphab_2)$. Now we can state \cite[Theorem 1.5]{galanis2015inapproximability}.\footnote{Technically, \cite[Theorem 1.5]{galanis2015inapproximability} demands the assumption $\NP\neq \RP$ but that is merely to exclude randomised algorithms, the reduction itself is deterministic.}

\begin{proposition}[{\cite{galanis2015inapproximability}}] \label{prop:GSV-hardness}
  Let $\Delta\geq 3$ be an integer, and suppose that $\BB$ is an ergodic interaction matrix  of an antiferromagnetic spin system. Suppose further that the dominant phases $(\alphab,\betab)$  satisfy $\alphab\neq\betab$, are permutation symmetric and Hessian dominant. Then, it is $\NP$-hard to approximate the partition function $Z_\BB(G)$ on $n$-vertex triangle-free $\Delta$-regular graphs $G$, even within a factor of $2^{cn}$ for a constant $c(\BB,\Delta)>0$.
\end{proposition}

The key ingredient in  the conditions of \Cref{prop:GSV-hardness} is the condition that $\alphab\neq \betab$; this enables a reduction in \cite{galanis2015inapproximability} to the $\MaxCut$ problem; the Hessian dominance and the permutation symmetry condition are more on the technical side, but is one of the main reasons that complicates the overall arguments (this was already prevalent in \cite{galanis2015inapproximability}).   

The main challenge to show our inapproximability results is to establish the conditions of \Cref{prop:GSV-hardness} for $\BB$ and the relevant range for $\Delta$, which is the scope of the following lemma.

\begin{lemma}  \label{lem:dominant-phase}
  Let $q\ge 4$ be even, $k\ge 2$, and $\Delta=5q^k+1$. Then the dominant phases of the spin system with interaction matrix $\BB$ (defined in \eqref{eq:BBdef}) satisfy the conditions of \Cref{prop:GSV-hardness}.
\end{lemma}

\Cref{thm:main-hardness} follows from \Cref{lem:graph-to-hypergraph}, \Cref{prop:GSV-hardness}, and \Cref{lem:dominant-phase}. It remains to analyse the dominant phases of $\BB$ and establish \Cref{lem:dominant-phase}, which is the focus of \Cref{sec:dominant-phase}.


\section{Analysis of the dominant phases}
\label{sec:dominant-phase}

In this section we analyze the dominant phase.
We will state the main lemmas in this section and in \Cref{sec:3-values}.
However, because the calculations are often very heavy,
many of the lemmas are not immediately proved.
The sections in which their proofs appear can be found in \Cref{tab:lemma-proof} at the end of \Cref{sec:3-values}.

Let $q,\Delta\geq 3$ be integers. To prove \Cref{lem:dominant-phase}, we need to analyse the function $\Psi_1$ from \eqref{eq:Psipsi}. The function $\Psi_1$ turns out to be inconvenient to work with, but there is a simpler surrogate function $\Phi$ from \cite{galanis2015inapproximability} that we can use. For vectors $\rb=\{R_i\}_{i\in [\oq]}$ and  $\cb=\{C_i\}_{i\in [\oq]}$ with nonnegative entries, let 
\begin{equation}\label{def:Phi}
\Phi(\rb,\cb):=\Delta\ln\frac{\transpose{\rb}\BB\cb}{\norm{\rb}_{p}\norm{\cb}_{p}}, \mbox{ where } p=\Delta/(\Delta-1).
\end{equation}
It is not hard to see that for the matrix $\BB$ defined in \eqref{eq:BBdef}, the critical points of $\Phi$ satisfy the following equations:\footnote{Here, and elsewhere, we use the notation $x_i\propto y_i$ for $i\in [\oq]$ to denote that $x_i=Ay_i$ for $i\in [\oq]$, for some arbitrary $A$.}
\begin{equation}\label{equ:recursion}
    \begin{aligned}
        R_0 &\propto t^d\Big(tC_0+\sum_{j\in [q]; j\neq i}C_i\Big)^d,\quad
        R_i \propto \Big(tC_0+\sum_{j\in [q]; j\neq i}C_i\Big)^d & \text{ for } i\in [q];\\
        C_0 &\propto t^d\Big(tR_0+\sum_{i\in [q]; i\neq j}R_i\Big)^d,\quad
        C_j \propto \Big(tR_0+\sum_{i\in [q]; i\neq j}R_i\Big)^d & \text{ for } j\in [q],
    \end{aligned}
\end{equation}
where $t=(q^k-q)^{1/\Delta}$ and $d\defeq\Delta-1$. 
The equations in \eqref{equ:recursion} are often called the ``tree recursions'', because they are the same as the recursion for marginal probabilities on an infinite $d$-ary tree.
Note that $1\le t\le 1.0312$ for any $q\ge 4$, $k\ge 2$ and $d\ge 5q^k$. The connection between the functions $\Psi_1$ and $\Phi$ is detailed in the following result from \cite{galanis2015inapproximability}, applied to our setting.

\begin{proposition}[{\cite[Theorem 4.1]{galanis2015inapproximability}}] \label{prop:gsv_4_1}
Let $q,\Delta\geq 3$ be integers, and let $p=\Delta/(\Delta-1)$. Then, the local maxima of $\Phi$ and $\Psi_1$ happen at critical points, i.e., there are no local maxima on the boundary. 
The transformation $(\rb,\cb)\mapsto (\alphab,\betab)$ given by $\alpha_i=R_i^{p}/\norm{\rb}^p_p$ and  $\beta_i=C_j^{p}/\norm{\cb}^p_p$ for $i\in[\oq]$ yields a one-to-one correspondence between the critical points of $\Phi$ and $\Psi_1$. Moreover, for the corresponding critical points $(\rb,\cb)$ and $(\alphab,\betab)$ it holds that $\Psi_1(\alphab,\betab)=\Phi(\rb,\cb)$.
\end{proposition} 

The function $\Phi$ is still multi-dimensional ($2q$), but fortunately we can reduce its dimensions significantly down to 11 by studying the structure of fixpoints to the system \eqref{equ:recursion}. A first observation is that $R_i<R_j$ implies $C_i>C_j$, and $R_i=R_j$ implies $C_i=C_j$, where $i,j\neq 0$. The next lemma is similar to \cite[Lemma 7.6]{galanis2015inapproximability}. 

\begin{lemma} \label{lem:fp-supp-count}
  Let $(R_0,R_1,\cdots,R_q,C_0,C_1,\cdots,C_q)$ be a positive fixpoint of \eqref{equ:recursion}.
  Then the number of distinct values in $\{R_i\}_{1\le i\le q}$ and $\{C_i\}_{1\le i\le q}$ is at most $3$.
\end{lemma}
\begin{proof}
  Let $R:=\sum_{i=1}^{q}R_i$ and $C:=\sum_{i=1}^{q}C_i$. 
  Suppose all variables are normalized so that $R_0+R=C_0+C=1$.
  Then for any $i\in[q]$, we have that
  \begin{align*}
    \frac{R_i}{R_0} & = t^{-d}\Big( \frac{(t-1)C_0+1-C_i}{(t-1)C_0+1} \Big)^d = t^{-d}\Big( \frac{(t-1)+C_0^{-1}-C_i/C_0}{(t-1)+C_0^{-1}} \Big)^d = t^{-d}\Big( 1-\frac{1}{t^dC'}\Big( 1-\frac{R_i}{R_0R'} \Big)^d \Big)^d,
  \end{align*}
  where $C'=(t-1)+C_0^{-1}$ and $R'=(t-1)+R_0^{-1}$.
  Let $x=\left( R_i/R_0 \right)^{1/d}$ and note that $x\in [0,1]$. Then the above equation becomes $f(x)=0$, where $f(x): = t^{-1}\Big(1-\tfrac{1}{t^dC'}\big( 1-\tfrac{x^d}{R'}\big)^d\Big)-x$. We have that
  \begin{align*}
    f'(x):=(g(x))^{d-1}-1, \mbox{ where } g(x):=\Big(\tfrac{d^2}{t^{d+1}R'C'}\Big)^{1/(d-1)}\Big( 1-\frac{x^d}{R'}\Big)x.
  \end{align*}
  Note that $g(x)>0$ on the interval $[0,1]$ because $x^d=\frac{R_i}{R_0} < (t-1)+R_0^{-1}=R'$. Using that $(g(x))^{d-1}-1=(g(x)-1)(g(x)^{d-2}+\hdots+1)$, we therefore obtain that the roots of $f'(x)=0$ can only come from the roots of $g(x)-1$, 
  which has at most two roots by the Descartes' rule of signs.  Hence $f'(x)$ changes its sign at most twice in the interval of $[0,1]$ and $f(x)$ has at most $3$ roots over $[0,1]$, showing that the $R_i$'s for $i\in [q]$ can only be supported on three different values. The statement for the $C_i$'s follows by an analogous argument. 
\end{proof}

The above lemma motivates the following definition. 

\begin{definition} \label{def:fixpoint_type}
  Let $(R_0,R_1,\cdots,R_q,C_0,C_1,\cdots,C_q)$ be a positive fixpoint. We call the fixpoint $m$-supported, if the number of distinct values in $\{R_i\}_{1\leq i\leq q}$ is $m$, where $m\in\{1,2,3\}$. We call the fixpoint is of type $(q_1,q_2,q_3)$ where $q_1+q_2+q_3=q$, if the multiplicities of different numbers in $\{R_i\}_{1\leq i\leq q}$ are $q_1,q_2,q_3$ respectively.\footnote{Any permutation over $q_1,q_2,q_3$ is considered equivalent. E.g., $(q/2,q/2,0)$ and $(q/2,0,q/2)$ are regarded as the same type.} In case that the fixpoint is $2$ or $1$-supported, let one or two of $q_i$'s take zero respectively. 
\end{definition}

From now on we may also abuse the notation $R_i$ (also $C_i$, $i=1,2,3$) by absorbing all the same values, and hence $R_1$ stands for the value that $q_1$ of $R$'s (except $R_0$) take, rather than the value of $R$ on the first index in the fixpoint. 

The main lemma of this section can be stated as follows. 
\begin{lemma}\label{lem:main}
  Suppose $q\geq 4$ is even, $k\geq 2$ and $d=5q^k$. The maximum of $\Psi_1$ over $(q_1,q_2,q_3)$-type fixpoints is attained uniquely, when $(q_1,q_2,q_3)=(q/2,q/2,0)$. 
\end{lemma}

We also need to prove that $2$-maximal triples $(q/2,q/2,0)$ yield unique $\mathbf{r}$ and $\mathbf{c}$ (up to scaling and permutation), and that the corresponding maxima are Hessian dominant. 
\begin{lemma} \label{lem:2_max_stable}
  Suppose $q\geq 4$ is even, $k\geq 2$ and $d\geq 5q^k$. The fixpoints of type $(q/2,q/2,0)$ are unique up to scaling and permutation symmetric. In addition, they are Hessian dominant maxima of $\Psi_1$.
\end{lemma}
Lemma~\ref{lem:dominant-phase} follows immediately by combining Lemmas~\ref{lem:main} and~\ref{lem:2_max_stable}
.

\subsection{Restricting to three values}\label{sec:3-values}

In order to prove Lemma \ref{lem:main}, we need to determine which type of fixpoints maximizes $\Psi_1$. By using Proposition \ref{prop:gsv_4_1}, the value of $\Psi_1$ corresponding to such a fixpoint in (\ref{equ:recursion}) can be given by the matrix norm (\ref{def:Phi}), which can be seen to be equal to
\begin{equation} \label{equ:phi_qrc_def}
\begin{aligned}
  &\overline{\Phi^S}(\mathbf{q},\mathbf{r},\mathbf{c}):=\\
  &(d+1)\ln\Big(R_0C_0t^2+\left(\mbox{$\sum_{i=1}^{3}$\,}q_iR_i\right)C_0t+\left(\mbox{$\sum_{i=1}^{3}$\,}q_iC_i\right)R_0t+\left(\mbox{$\sum_{i=1}^{3}$\,}q_iR_i\right)\left(\mbox{$\sum_{i=1}^{3}$\,}q_iC_i\right)-\left(\mbox{$\sum_{i=1}^{3}$\,}q_iR_iC_i\right)\Big)\\
  &\hskip 2cm-d\ln\left(R_0^{(d+1)/d}+\mbox{$\sum_{i=1}^{3}$\,}q_iR_i^{(d+1)/d}\right)-d\ln\left(C_0^{(d+1)/d}+\mbox{$\sum_{i=1}^{3}$\,}q_iC_i^{(d+1)/d}\right).
\end{aligned}
\end{equation}
Where we define the vector $\mathbf{r}=(R_0,R_1,R_2,R_3)$ and $\mathbf{c}=(C_0,C_1,C_2,C_3)$. It is worth noting that this function is scale-free with respect to $\mathbf{r}$ and $\mathbf{c}$, as this property will be used intensively in our later proofs. 

The discrete optimization of (\ref{equ:phi_qrc_def}) over all fixpoints of the tree recursion (\ref{equ:recursion}) is difficult to cope with. Instead, we then try to maximize (\ref{equ:phi_qrc_def}) over all nonnegative $\mathbf{q}$ and $\sum_{i=1}^{3}q_i=q$, wishing the maximum to be taken at integer $\mathbf{q}$. This is the main reason such approach can only deal with even $q$. 

For all $\mathbf{q}=(q_1,q_2,q_3)$ with $q_1+q_2+q_3=q, q_i\geq 0$, define
\begin{equation} \label{equ:phi_q_def}
  \overline{\Phi}(\mathbf{q}):=\max_{\mathbf{r},\mathbf{c}}\overline{\Phi^S}(\mathbf{q},\mathbf{r},\mathbf{c})
\end{equation}
where the maximum is taken over $\mathbf{r}=(R_0,R_1,R_2,R_3), \mathbf{c}=(C_0,C_1,C_2,C_3)$ satisfying
\begin{equation} \label{equ:phi_cond}
  \begin{gathered}
R_0C_0t^2+\left(\mbox{$\sum_{i=1}^{3}$\,}q_iR_i\right)C_0t+\left(\mbox{$\sum_{i=1}^{3}$\,}q_iC_i\right)R_0t+\left(\mbox{$\sum_{i=1}^{3}$\,}q_iR_i\right)\left(\mbox{$\sum_{i=1}^{3}$\,}q_iC_i\right)-\left(\mbox{$\sum_{i=1}^{3}$\,}q_iR_iC_i\right)>0,\\
    R_i, C_i\geq 0, i=0,1,2,3.
  \end{gathered}
\end{equation}

Our first step is to verify the maximum in (\ref{equ:phi_q_def}) is well defined, and moreover, the maximum in $\max_{\mathbf{q}}\overline{\Phi}(\mathbf{q})$ can also be taken. This is formalized by the next lemma.
\begin{lemma}[{\cite[Lemma 7.10]{galanis2015inapproximability}}] \label{lem:max_well_defined}
The maximum in \eqref{equ:phi_q_def} is well-defined. In addition, $\max_{\mathbf{q}}\overline{\Phi}(\mathbf{q})$ can be attained in the region where $q_1+q_2+q_3=q, q_i\geq 0$. 
\end{lemma}
\begin{proof}
The argument is verbatim the same as in \cite{galanis2015inapproximability}, the only difference is that the function has slightly different form, but still accounts for the relevant parameters $q_1,q_2,q_3$.
\end{proof}

The next trouble we may encounter later is that we are now dealing with all possible $\mathbf{r},\mathbf{c}$ conditioned on (\ref{equ:phi_cond}), instead of just fixpoints of (\ref{equ:recursion}). The good news is that, in contrast to  \cite{galanis2015inapproximability}, we can rule out fairly easily that the maximizer in \eqref{equ:phi_q_def} is at the boundary.
\begin{lemma} \label{lem:bad_triple}
For all triples ${\bm q}=(q_1,q_2,q_3)$, any maximizer in \eqref{equ:phi_q_def} satisfies (a) $R_0,C_0>0$, (b) for any $i$ such that $q_i>0$, it holds that $R_i,C_i>0$, and (c) for distinct $i,j$ such that $q_i,q_j>0$, it holds that  $R_i=R_j$ if and only if $C_i=C_j$. 
\end{lemma}
The problem in \cite{galanis2015inapproximability} that also appears in our setting is that it might be that $q_i,q_j>0$, but $R_i=R_j$ and $C_i=C_j$.  For example, imagine we are now strengthening the restriction \eqref{equ:phi_cond} by adding $R_1=R_2$ and $C_1=C_2$. Then $\overline{\Phi}(q_1+q_2,0,q_3)\leq\overline{\Phi}(q_1,q_2,q_3)$. Such degenerate case makes it difficult to compare between different $\mathbf{q}$ triples because the equality can be taken. This motivates the next definition. 
\begin{definition} \label{def:maximal}
  Let $m=2,3$. A triple $\mathbf{q}$ is called \emph{$m$-maximal}, if exactly $m$ $q_i$'s in $\mathbf{q}$ are non-zero, and there exists $\mathbf{r},\mathbf{c}$ maximizing (\ref{equ:phi_q_def}) such that, $q_i,q_j>0$ and $i\neq j$ imply that $R_i\neq R_j$ and $C_i\neq C_j$. 
  We also call $\mathbf{q}$ \emph{maximal} if it is either $2$- or $3$-maximal.
\end{definition}

Now we connect $m$-maximal triples with fixpoints in \eqref{equ:recursion}. 
\begin{lemma} \label{lem:maximal_to_fixpoint}
  Suppose a triple $\mathbf{q}$ is $m$-maximal. Then there exists $\mathbf{r},\mathbf{c}$ achieving the maximum in \eqref{equ:phi_q_def} and specifying an $m$-supported fixpoint of tree recursion \eqref{equ:recursion} of type $\mathbf{q}$. 
\end{lemma}

For $2$ and $3$-maximal triples, the key is the next lemma. 
\begin{lemma} \label{lem:maximal_max}
  Suppose $q\geq 4$ is even. Then the following statements hold:
  \begin{itemize}
      \item[(a)] There does not exist any $3$-maximal triple that maximizes \eqref{equ:phi_q_def}.
      \item[(b)] The only possibility of a $2$-maximal triple to maximize \eqref{equ:phi_q_def} is $(q/2,q/2,0)$ or its permutations, with $R_i/R_j=C_j/C_i$, where $i\neq j$ are the two indices such that $q_i,q_j=q/2$.
  \end{itemize}
\end{lemma}

The above  lemma is not yet enough to finish the proof of Lemma~\ref{lem:main} because we have to rule out degenerate cases of all triples, i.e., the triple $(q,0,0)$. This is the main difference with the colour-symmetric setting of \cite{galanis2015inapproximability}. Instead, we have the special colour corresponding to ($R_0,C_0$), which makes the system behave like a $2$-spin system when all ``pure'' colours take the same fraction. What is worse is that, it is possible for the $2$-spin system to have three fixpoints (two of them being symmetric), when the tree recursion lies in the so-called ``non-uniqueness'' region (see Section \ref{sec:q00_stable}). Therefore, we need to discuss such fixpoints by two different cases. 

Before continuing the discussion, let us state another useful result from \cite{galanis2015inapproximability}. 
A fixpoint $x$ of a mapping $f$ is \emph{Jacobian stable} if the Jacobian of $f$ at $x$ has spectral radius less than $1$.
\begin{proposition}[{\cite[Theorem 4.2]{galanis2015inapproximability}}] \label{thm:gsv_4_2}
  A fixpoint of the tree recursion \eqref{equ:recursion} is Jacobian stable if and only if it corresponds to a Hessian dominant local maximum of $\Psi_1$. 
\end{proposition}

The first kind of fixpoints satisfy $R_0/R_1\neq C_0/C_1$. As stated in the next lemma, such a fixpoint is Jacobian stable, and hence it is a possible candidate to be the maximizer in $\max_{\mathbf{q}}\overline{\Phi}(\mathbf{q})$. Though the proof of stability is not necessary for our main theorem, we still leave it here for future references. 
\begin{lemma} \label{lem:q00_inequal_stable}
  Suppose $d\geq 5q^k$. The fixpoint corresponding to triple $(q,0,0)$ and $R_0/R_1\neq C_0/C_1$ is unique up to scaling and swapping $R$ and $C$. Moreover, it is Jacobian stable.
\end{lemma}
For the reason above, we can only go through a very detailed calculation to rule out this case. The equality in $d=5q^k$ from the next lemma is for the sake of simplification in calculation. 
\begin{lemma} \label{lem:q00_not_max}
  Suppose $d=5q^k$. Any fixpoint corresponding to triple $(q,0,0)$ and $R_0/R_1\neq C_0/C_1$ does not maximize (\ref{equ:phi_q_def}). 
\end{lemma}
On the other hand, when $R_0/R_1=C_0/C_1$, things become easier as such fixpoints are not Jacobian stable. 
Thus, by \Cref{thm:gsv_4_2}, these fixpoints do not correspond to local maxima of $\Psi_1$. 
\begin{lemma} \label{lem:q00_equal_unstable}
  Suppose $d\geq 5q^k$. Any fixpoint corresponding to triple $(q,0,0)$ and $R_0/R_1=C_0/C_1$ is Jacobian unstable.
\end{lemma}

Now we are ready to prove Lemma \ref{lem:main}, which given the above ingredients can be done by following closely a related argument in \cite{galanis2015inapproximability}.
The main complicacy in the proof is that when we find a maximizer $\mathbf{q}$ of $\overline{\Phi}(\mathbf{q})$,
the corresponding $\mathbf{r}$ (or $\mathbf{c}$) values are not necessarily distinct.
We need to carefully rule out these degenerate cases.
\begin{proof}[Proof of Lemma \ref{lem:main}]
  Denote $MAX:=\max_{\mathbf{q}}\overline{\Phi}(\mathbf{q})$. We first claim that $MAX$ is attained at $\hat{\mathbf{q}}=(q/2,q/2,0)$, and $\hat{\mathbf{q}}$ is maximal. Assuming the claim, Lemma \ref{lem:maximal_to_fixpoint} yields that there exist $\hat{\mathbf{r}},\hat{\mathbf{c}}$ with $\overline{\Phi}(\hat{\mathbf{q}})=\overline{\Phi^S}(\hat{\mathbf{q}},\hat{\mathbf{r}},\hat{\mathbf{c}})$,  specifying a $(q/2,q/2,0)$-type fixpoint of the tree recursion \eqref{equ:recursion}. Hence $MAX=\max\Psi_1$. To show that $\hat{\mathbf{q}}$ is the unique type of  fixpoint achieving the maximum of $\Psi_1$, consider an arbitrary  $\mathbf{q}^*$-type fixpoint achieving the maximum of $\Psi_1$, say $(\mathbf{r}^*,\mathbf{c}^*)$.  Then $\mathbf{q}^*$ must also achieve the maximum in $\max_{\mathbf{q}}\overline{\Phi}(\mathbf{q})$. By Lemma \ref{lem:q00_not_max}, $\mathbf{q}^*\neq(q,0,0)$ and hence it is maximal according to Definition \ref{def:maximal} (using $(\mathbf{r}^*,\mathbf{c}^*)$ as the maximizers; Recall Definition \ref{def:fixpoint_type} that $R_i\neq R_j,C_i\neq C_j$ for $i\neq j$ and $q_i,q_j>0$). Therefore we can apply Lemma \ref{lem:maximal_max} and obtain that $\mathbf{q}^*=\hat{\mathbf{q}}$.  

  It remains to prove the claim above. Let $\mathbf{q}^*$ be any maximizer of $\max_{\mathbf{q}}\overline{\Phi}(\mathbf{q})$.
  \begin{itemize}
    \item[(1)] $\mathbf{q}^*$ has at least two positive entries. This is a consequence of Lemmas \ref{lem:q00_not_max} and  \ref{lem:q00_equal_unstable} (after using \Cref{thm:gsv_4_2}). 
    \item[(2)] In case $\mathbf{q}^*$ has exactly two positive entries, then $\mathbf{q}^*$ must be maximal. Otherwise, suppose $\mathbf{q}^*=(q_1,q_2,0)$ and the maximizer in (\ref{equ:phi_q_def}) is achieved at $\mathbf{r}^*,\mathbf{c}^*$ where $R_1=R_2$ or $C_1=C_2$. By Lemma~\ref{lem:bad_triple} (c), both equalities are true and hence $\overline{\Phi}(\mathbf{q}^*)=\overline{\Phi}((q,0,0))$, contradicting item (1). 
    \item[(3)] In case $\mathbf{q}^*$ has exactly two positive entries, it must holds that $\mathbf{q}^*=\hat{\mathbf{q}}$. This is from item (2), and Lemma~\ref{lem:maximal_max} (b). 
    \item[(4)] If $\mathbf{q}^*$ has all positive entries, then it cannot be $3$-maximal. This is from Lemma \ref{lem:maximal_max} (a). 
    \item[(5)] If $\mathbf{q}^*$ has all positive entries, then $\overline{\Phi}(\mathbf{q}^*)=\overline{\Phi}(\hat{\mathbf{q}})$. This can be proved by the following argument. Let $\mathbf{r}^*,\mathbf{c}^*$ be the maximizer corresponding to $\mathbf{q}^*$. By item (4), $\mathbf{q}^*$ is not $3$-maximal, and using the argument of item (2), there exist distinct $i, j\geq 1$ such that $R_i=R_j$ and $C_i=C_j$ in $\mathbf{r}^*,\mathbf{c}^*$. Let $k\geq 1$ be the remaining index. 
    \begin{itemize}
      \item If $R_i=R_j=R_k$, then by Lemma~\ref{lem:bad_triple} (c),  $C_i=C_j=C_k$, and hence $\overline{\Phi}(\mathbf{q}^*)=\overline{\Phi}(q,0,0)$, contradicting item (1).
      \item If $C_i=C_j=C_k$, then by Lemma~\ref{lem:bad_triple} (c),  $R_i=R_j=R_k$, and hence $\overline{\Phi}(\mathbf{q}^*)=\overline{\Phi}(q,0,0)$, contradicting item (1).
      \item If $R_i\neq R_k$ and $C_i\neq C_k$, we can   ``merge'' the indices $i,j$ to get a new triple $\mathbf{q}':=(q_i+q_j,q_k,0)$. Let $\rb':=(R_0,R_i,R_k,0),\cb':=(R_0,C_i,C_k,0)$. Then \[\overline{\Phi}(\mathbf{q}^*)=\overline{\Phi^S}(\mathbf{q}^*,\mathbf{r}^*,\mathbf{c}^*)=\overline{\Phi^S}(\mathbf{q}',\mathbf{r}',\mathbf{c}')\leq \overline{\Phi}(\mathbf{q}').\] This means that $\mathbf{q}'$ is also a maximizer of $\max_{\mathbf{q}}\overline{\Phi}(\mathbf{q})$ since $\mathbf{q}^*$ is a maximizer. However, $\mathbf{q}'$ has exactly two positive entries. Hence by item (3), $\overline{\Phi}(\mathbf{q}^*)=\overline{\Phi}(\mathbf{q}')=\overline{\Phi}(\hat{\mathbf{q}})$ . 
    \end{itemize} 
  \end{itemize}
  The above arguments imply that for any maximizer $\mathbf{q}^*$, it holds that $\overline{\Phi}(\mathbf{q}^*)=\overline{\Phi}(\hat{\mathbf{q}})$, which means that $\hat{\mathbf{q}}$ is indeed a maximizer. This also indicates all items above apply to $\mathbf{q}^*=\hat{\mathbf{q}}$, and from item (3), we obtain that $\hat{\mathbf{q}}$ is $2$-maximal. This concludes the proof. 
\end{proof}

Before diving into the proofs of all the lemmas above, we want to mention the following observation. The partial derivatives $\partial\overline{\Phi^S}/\partial q_i$, conditioned on $\mathbf{r}$ and $\mathbf{c}$ achieving the maximum in (\ref{equ:phi_q_def}), can be written as follows. (Note that it applies to all triples $\mathbf{q}$, including non-maximal ones.) Based on these partial derivatives, we can argue the non-optimality by perturbing $q_i$'s.
\begin{lemma} \label{lem:dphi_dqi_expr}
Suppose $\mathbf{r}$, $\mathbf{c}$ achieve the maximum in (\ref{equ:phi_q_def}). Then for any $i\in\{1,2,3\}$ such that $q_i>0$, it holds that
\[
  \frac{\partial\overline{\Phi^S}}{\partial q_i}
  =
  \frac{R_iC_0t + R_0C_it + (d-1)R_iC_i+R_i\left(\sum_{j=1}^{3}C_jq_j\right)+C_i\left(\sum_{j=1}^{3}R_jq_j\right)}
  {R_0C_0t^2+\left(\sum_{j=1}^{3}C_jq_j\right)R_0t+\left(\sum_{j=1}^{3}R_jq_j\right)C_0t+\left(\sum_{j=1}^{3}R_jq_j\right)\left(\sum_{j=1}^{3}C_jq_j\right)-\left(\sum_{j=1}^{3}R_jC_jq_j\right)}.
\]
Moreover, if there exists $i,j$ such that $q_i,q_j>0$ and $i\neq j$ and satisfies $\partial\overline{\Phi^S}/\partial q_i-\partial\overline{\Phi^S}/\partial q_j\neq 0$, then the maximum in (\ref{equ:phi_q_def}) is not achieved. 
\end{lemma}

Unproved propositions and lemmas in this subsection can be found later. We make a list of where they are proved. 
\begin{table}
\begin{center}
\begin{tabular}{c|c}
\hline
Proposition/Lemma & Section\\
\hline
Lemma \ref{lem:maximal_max} & Section \ref{sec:maximal}\\
Lemma \ref{lem:2_max_stable} & Section \ref{sec:2max_stable}\\
Lemma \ref{lem:q00_inequal_stable}, Lemma \ref{lem:q00_equal_unstable} & Section \ref{sec:q00_stable}\\
Lemma \ref{lem:q00_not_max} & Section \ref{sec:q00_not_max}\\
Lemma \ref{lem:maximal_to_fixpoint}, Lemma \ref{lem:dphi_dqi_expr} & Section \ref{sec:proof_d_expr_equations}\\
Lemma \ref{lem:bad_triple} & Section \ref{sec:bad_triple}\\
\hline
\end{tabular}
\end{center}
\caption{The sections where the lemmas are proved}
\label{tab:lemma-proof}
\end{table}

\subsection{\texorpdfstring{$2,3$}{2,3}-maximal Triples} \label{sec:maximal}
Let $\mathbf{q}$ be a maximal triple and let $I=\{i\mid q_i>0\}$. From Lemma~\ref{lem:bad_triple} (a) and (b), by taking partial derivatives of $\overline{\Phi^S}$ with respect to non-zero $R_i$ and $C_i$'s and setting them to 0, we get that the maximizer of $\overline{\Phi^S}$ satisfies
\begin{align}
  R_0^{1/d}&\propto C_0t^2+(q_1C_1+q_2C_2+q_3C_3)t,\quad
  &R_i^{1/d}\propto C_0t+q_1C_1+q_2C_2+q_3C_3-C_i \mbox{ for $i\in I$};\label{equ:phi_deri_zero_R}\\
    C_0^{1/d}&\propto R_0t^2+(q_1R_1+q_2R_2+q_3R_3)t,\quad
  &C_i^{1/d}\propto R_0t+q_1R_1+q_2R_2+q_3R_3-R_i \mbox{ for $i\in I$}. \label{equ:phi_deri_zero_C}
\end{align}
First assume $\mathbf{q}$ is $3$-maximal, for any $i\neq j$ it holds that $R_i\neq R_j$ and $C_i\neq C_j$. From Lemma~\ref{lem:bad_triple} (a) and (b), we may assume the following strict ordering
\begin{equation} \label{equ:3_max_order}
\begin{aligned}
  R_1>R_2>R_3>0 && \text{and} && 0<C_1<C_2<C_3.
\end{aligned}
\end{equation}

\begin{lemma} \label{lem:3_max_maximizer}
  Suppose $R_i$'s and $C_i$'s satisfy (\ref{equ:phi_deri_zero_R}), (\ref{equ:phi_deri_zero_C}) and (\ref{equ:3_max_order}). We have the following:
  \begin{itemize}
    \item[(a)] If $R_1/R_3\neq C_3/C_1$, then $\partial\overline{\Phi^S}/\partial q_1-\partial\overline{\Phi^S}/\partial q_3\neq 0$.
    \item[(b)] If $R_1/R_3=C_3/C_1$, then $\partial\overline{\Phi^S}/\partial q_1-\partial\overline{\Phi^S}/\partial q_2\neq 0$.
  \end{itemize}
\end{lemma}
For the sake of convenience, we further set
\[
  r_0^d:=R_0/R_3, r_1^d:=R_1/R_3, r_2^d:=R_2/R_3, \text{ and } c_0^d:=C_0/C_1, c_2^d:=C_2/C_1, c_3^d:=C_3/C_1.
\]
which means
\begin{equation} \label{equ:3_max_order_small}
  \begin{aligned}
    r_1>r_2>1 && \text{and} && c_3>c_2>1.
  \end{aligned}
\end{equation}
We will need these notations in later sections too.
With them, from \eqref{equ:phi_deri_zero_R} and \eqref{equ:phi_deri_zero_C}, we obtain that
\begin{equation}\label{equ:small_r_0_c_0}
r_0=\frac{c_0^d t^2+(q_1+q_2c_2^d+q_3c_3^d)t}{c_0^d t+q_1+q_2c_2^d+(q_3-1)c_3^d},\qquad 
c_0=\frac{r_0^d t^2+(q_1r_1^d+q_2r_2^d+q_3)t}{r_0^d t+(q_1-1)r_1^d+q_2r_2^d+q_3}.
\end{equation}
\begin{equation}\label{equ:small_r_1_c_3}
r_1=\frac{c_0^d t+q_1-1+q_2c_2^d+q_3c_3^d}{c_0^d t+q_1+q_2c_2^d+(q_3-1)c_3^d},\qquad
c_3=\frac{r_0^d t+q_1r_1^d+q_2r_2^d+q_3-1}{r_0^d t+(q_1-1)r_1^d+q_2r_2^d+q_3},
\end{equation}
\begin{proof}[Proof of Lemma~\ref{lem:3_max_maximizer}]
From \eqref{equ:phi_deri_zero_R} and \eqref{equ:phi_deri_zero_C}, we get
  \begin{equation}\label{equ:small_r_1_over_r_2_to_r_2} 
    \frac{r_1-1}{r_2-1}=\frac{c_3^d-1}{c_3^d-c_2^d}, \mbox{ yielding that } r_2=\frac{r_1c_3^d-1-c_2^d(r_1-1)}{c_3^d-1}
  \end{equation}
 Similarly, we obtain that
  \begin{equation} \label{equ:small_c_3_over_c_2}
    \frac{c_3-1}{c_2-1}=\frac{r_1^d-1}{r_1^d-r_2^d}, \mbox{ yielding that } r_2^d=\frac{r_1^dc_3-1-c_2(r_1^d-1)}{c_3-1},
  \end{equation}
  
From \eqref{equ:phi_deri_zero_R} and \eqref{equ:phi_deri_zero_C}, we have that $r_2=\frac{c_0^d t+q_1+(q_2-1)c_2^d+q_3c_3^d}{c_0^d t+q_1+q_2c_2^d+(q_3-1)c_3^d}$ which combined  with \eqref{equ:small_r_1_over_r_2_to_r_2} gives 
that 
\begin{equation}\label{eq:sys1sys}
\frac{c_0^d t+q_1+(q_2-1)c_2^d+q_3c_3^d}{c_0^d t+q_1+q_2c_2^d+(q_3-1)c_3^d}=\frac{r_1c_3^d-1-c_2^d(r_1-1)}{c_3^d-1}.
\end{equation}
Symmetrically we obtain that 
  \begin{equation}\label{eq:sys2sys}
    \frac{r_0^d t+q_1r_1^d+(q_2-1)c_2^d+q_3}{r_0^d t+(q_1-1)r_1^d+q_2c_2^d+q_3}
    =\frac{r_1^dc_3-1-r_2^d(c_3-1)}{r_1^d-1}.
  \end{equation}
  We can view \eqref{eq:sys1sys} and \eqref{eq:sys2sys} as a linear system in $q_1$ and $q_3$ after clearing the denominators, which yields that 
\begin{align}
q_1\cdot(r_1^dc_3^d-1)&=c_0^d t+q_2c_2^d+\frac{1-r_1c_3^d}{r_1-1}-c_3^d\left(r_0^d t+q_2r_2^d+\frac{1-c_3r_1^d}{c_3-1}\right),\label{equ:q_1_q_3_system_sol_q_1}\\
  q_3\cdot(r_1^dc_3^d-1)&=r_0^d t+q_2r_2^d+\frac{1-r_1^dc_3}{c_3-1}-r_1^d\left(c_0^d t+q_2c_2^d+\frac{1-r_1c_3^d}{r_1-1}\right),\label{equ:q_1_q_3_system_sol_q_3}
 \end{align}
  
From \eqref{equ:phi_deri_zero_R} and \eqref{equ:phi_deri_zero_C}, we also obtain that
\begin{equation}\label{equ:r0dt}
  r_0^dt=\frac{r_1^d-1}{c_3-1}-(q_1-1)r_1^d-q_2r_2^d-q_3,\qquad
  c_0^dt=\frac{c_3^d-1}{r_1-1}-q_1-q_2c_2^d-(q_3-1)c_3^d.
\end{equation}
We can now show the following:
\begin{equation} \label{lem:r_1_eq_c_3}
\mbox{ if $r_1=c_3$, then (i) $r_2=c_2$, (ii) $q_1=q_3$, and (iii) $r_0=c_0$.}
\end{equation}
The proof of (i) in \eqref{lem:r_1_eq_c_3} is identical to that in \cite[Lemma 7.20]{galanis2015inapproximability} using the expressions for $r_2,r_2^d$ in \eqref{equ:small_r_1_over_r_2_to_r_2} and \eqref {equ:small_c_3_over_c_2}, respectively. From (i) and the assumption that $r_1=c_3$, we obtain from   \eqref{equ:q_1_q_3_system_sol_q_1} and \eqref{equ:q_1_q_3_system_sol_q_3} that
\begin{equation} \label{equ:q_3_minus_q_1}
  q_3-q_1=(r_0^d -c_0^d) \cdot\frac{t}{r_1^{d}-1}. 
\end{equation}
Furthermore, the equations in \eqref{equ:small_r_0_c_0}  can also be regarded as a linear system in $q_1$ and $q_3$. Using the assumption $r_1=c_3$ and $r_2=c_2$, we obtain that 
\begin{equation*}
  q_3-q_1=\frac{t\left[(r_0^d-c_0^d)t^2+(r_0c_0^d-c_0r_0^d+c_0^{1+d}-r_0^{1+d})t+(r_0^d-c_0^d)r_0c_0-(r_0-c_0)r_1^d\right]}{(r_0-t)(c_0-t)(r_1^d-1)}. 
\end{equation*}
which, together with (\ref{equ:q_3_minus_q_1}), implies $r_0=c_0$ and hence $q_1=q_3$.
This finishes proving \eqref{lem:r_1_eq_c_3}.

We are now ready to give the proof of Lemma \ref{lem:3_max_maximizer}. For part (a),  Lemma \ref{lem:dphi_dqi_expr} yields that
\begin{equation*}
\begin{aligned}
  \frac{\partial\overline{\Phi^S}}{\partial q_1}-\frac{\partial\overline{\Phi^S}}{\partial q_3}=\frac{1}{S}\times
  \bigg[&(r_1^d-1)(c_0^dt+q_1+c_2^dq_2)-(c_3^d-1)(r_0^dt+q_3+r_2^dq_2)\\
  &+(d-1)(r_1^d-c_3^d)+r_1^dc_3^d(q_3-q_1)+r_1^dq_1-c_3^dq_3 \bigg]. 
\end{aligned}
\end{equation*}
where $S>0$. Then plug in the expression of $q_1$ and $q_3$ in (\ref{equ:q_1_q_3_system_sol_q_1}) and (\ref{equ:q_1_q_3_system_sol_q_3}), we get
\begin{equation} \label{equ:dq1_dq3}
  \frac{\partial\overline{\Phi^S}}{\partial q_1}-\frac{\partial\overline{\Phi^S}}{\partial q_3}=-\frac{g(r_1,c_3)}{S}
\end{equation}
where $
g(r_1,c_3):=(r_1-c_3)(r_1^d-1)(c_3^d-1)-d(r_1-1)(c_3-1)(r_1^d-c_3^d)$. This quantity was shown to have the same sign as $r_1-c_3$ (see Equation (123) in the proof of Lemma 7.19 in \cite{galanis2015inapproximability}), and specifically, non-zero when $r_1\neq c_3$, concluding part (a). 

Now we prove part (b) of Lemma \ref{lem:3_max_maximizer}. From \eqref{lem:r_1_eq_c_3}, the assumption $r_1=c_3$ implies $r_2=c_2$, $q_1=q_3$ and $r_0=c_0$. Applying Lemma \ref{lem:dphi_dqi_expr} based on these, we get
\begin{align*}
  \frac{\partial\overline{\Phi^S}}{\partial q_1}-\frac{\partial\overline{\Phi^S}}{\partial q_2}&=q_1(1+r_1^d)(1+r_1^d-2r_2^d)+r_2^d(q_2-(2q_2+d-1)r_2^d-2r_0^dt)+r_0^dt+r_1^d(r_0^dt+q_2r_2^d+d-1)\\
  &=-\frac{(d-1)(r_1-1)r_2^{2d}+2(r_1^{d+1}-1)r_2^d-(r_1^{2d+1}+dr_1^{d+1}-dr_1^d-1)}{r_1-1},
\end{align*}
where in the second line we use (\ref{equ:r0dt}). This quantity was shown to be non-zero in the proof of Lemma 7.19 in \cite{galanis2015inapproximability} (from Equation (124) onwards) under \eqref{equ:small_r_1_over_r_2_to_r_2},  concluding part (b). 
\end{proof}

Now we assume $\mathbf{q}=(q_1,q_2,q_3)$ is a $2$-maximal  triple, and assume $q_2=0$ without loss of generality. The result here is analogous to Lemma \ref{lem:3_max_maximizer} (a). 
\begin{lemma} \label{lem:2_max_maximizer}
  Under the assumption that $q_2=0$, suppose $R_i$'s and $C_i$'s ($i\neq 2$) satisfy (\ref{equ:phi_deri_zero_R}), (\ref{equ:phi_deri_zero_C}) and (\ref{equ:3_max_order}). For any $q_1,q_3>0$, it holds that $\partial\overline{\Phi^S}/\partial q_1-\partial\overline{\Phi^S}/\partial q_3\neq 0$, unless $q_1=q_3$ and $R_1/R_3=C_3/C_1$. 
\end{lemma}

\begin{proof}
First, note that the values of $R_2$ and $C_2$ do not affect the value of derivatives $\partial\overline{\Phi^S}/\partial q_1$ and $\partial\overline{\Phi^S}/\partial q_3$ when $q_2=0$. In addition, the expressions of $q_1$ and $q_3$ in (\ref{equ:q_1_q_3_system_sol_q_1}) and (\ref{equ:q_1_q_3_system_sol_q_3}) still hold for $q_2=0$. Therefore, one can carry out the proof of Lemma \ref{lem:3_max_maximizer} (a) once again for this case, showing $\partial\overline{\Phi^S}/\partial q_1-\partial\overline{\Phi^S}/\partial q_3=0$ only when $R_1/R_3=C_3/C_1$. Assuming this, one can show $q_1=q_3$ by going through the proof of \eqref{lem:r_1_eq_c_3}. 
\end{proof}

We conclude this subsection with Lemma \ref{lem:maximal_max}. 
\begin{proof}[Proof of Lemma \ref{lem:maximal_max}]
This comes after Lemma \ref{lem:3_max_maximizer}, Lemma \ref{lem:2_max_maximizer} and the second part of Lemma \ref{lem:dphi_dqi_expr}. 
\end{proof}

\subsection{Stability of Maximal \texorpdfstring{$(q/2,q/2,0)$}{(q/2,q/2,0)} Fixpoints} \label{sec:2max_stable}

In the next two subsections, we focus on the (in)stability of candidate fixpoints that may maximize $\Psi_1$. The condition of Jacobian stability is given by the following Lemma. 
\begin{lemma}[{cf. \cite[{Lemma 4.16}]{galanis2015inapproximability}}] \label{lem:gsv_4_16}
  Suppose $(R_0,R_1,\cdots,R_q,C_0,C_1,\cdots,C_q)$ is a fixpoint of the tree recursion (\ref{equ:recursion}). Let $\alpha_i:=\sum_{j=0}^{q}B_{ij}R_iC_j$ and $\beta_j:=\sum_{i=0}^{q}B_{ij}R_iC_j$. Define the matrix ${\bm A}:=(a_{ij})_{0\leq i,j\leq q}$ as $a_{ij}=B_{ij}R_iC_j/\sqrt{\alpha_i\beta_j}$, and the matrix ${\bm L}:=\left[\begin{smallmatrix} 0&{\bm A}\\{\bm A}^{\top}&0 \end{smallmatrix}\right]$. Then ${\bm L}$ has symmetric real spectrum (symmetry means if $a$ is an eigenvalue then so is $-a$), and $\pm 1$ is a pair of its eigenvalues. The condition for the fixpoint to be stable is that the second largest eigenvalue of ${\bm L}$ is less than $1/d$. 
\end{lemma}
We will also need the following lemma which is proved in Section~\ref{sec:proof_unique_root}.
\begin{lemma} \label{lem:r1_unique}
For any $q\geq 4, k\geq 2$ and $d\geq 3q^k$, the function
\begin{equation*}
h(x):=\left(\frac{x^{d+1}-1}{x^d-1}\right)^{d}t^{d+1}-\frac{x^d-1}{x-1}+q'+(q'-1)x^d
\end{equation*}
has exactly one root in the region $(1,\infty)$. 
\end{lemma}
We are now ready to prove Lemma \ref{lem:2_max_stable}.
\begin{proof}[Proof of Lemma \ref{lem:2_max_stable}]
Define $q':=q/2$. We first prove the uniqueness of $2$-maximal $(q',q',0)$ fixpoint (up to scaling). According to the proof of Lemma \ref{lem:2_max_maximizer}, fixpoints of type $(q', q', 0)$ maximize $\overline{\Phi}$ only when $r_1=c_3$. Now denote $x:=r_1=c_3$. To prove the first part of this lemma, we show there exists exactly one possible $x>1$ when $d\geq 3q^k$. By \eqref{equ:small_r_0_c_0} and \eqref{equ:small_r_1_c_3}, we get
\[
  \frac{r_0/t-r_1}{r_1-1}=\frac{1}{c_3^d-1}.
\]
Combining this with (\ref{equ:r0dt}), $x>1$ satisfies $h(x)=0$, where
\begin{equation*}
h(x):=\left(\frac{x^{d+1}-1}{x^d-1}\right)^{d}t^{d+1}-\frac{x^d-1}{x-1}+q'+(q'-1)x^d.
\end{equation*}
By Lemma \ref{lem:r1_unique}, $h(x)$ has exactly one root $x>1$. 

Next, we construct the matrices ${\bm A}$ and ${\bm L}$. Note that both matrices are scale-free with respect to $R_i$ and $C_i$. Directly plug in the formula in Lemma \ref{lem:gsv_4_16} to get
\[
{\bm A}:=\left[\begin{array}{ccc} c^2 & bc\transpose{\ones}  & ac\transpose{\ones} \\ ac\ones & ab \Jb& a^2\Jb'\\ bc\ones & b^2\Jb'& ab \Jb\end{array}\right].
\]
where $a:=\sqrt{x^{d-1}\frac{x-1}{x^d-1}}$, $b:=\sqrt{\frac{x-1}{x^d-1}}$ and $c:=\sqrt{\frac{x^{d+1}-1-q'(x-1)(x^d+1)}{x^{d+1}-1}}$, $\Jb$ is the $q'\times q'$ matrix with zeros on the diagonal and ones elsewhere, $\Jb'$ is the $q'\times q'$ matrix with ones everywhere, and $\ones$ is the $q'\times 1$ matrix. The eigenvalues of ${\bm L}=\left[\begin{smallmatrix} 0&{\bm A}\\{\bm A}^{\top}&0 \end{smallmatrix}\right]$ consist of $\pm ab$ (each by multiplicity $q-2$) and $\pm \lambda_1,\pm \lambda_2,\pm \lambda_3$, where $\lambda_1,\lambda_2,\lambda_3$ are the zeros of the following cubic function
\[
f(z)=z^3-(q'a^2+q'b^2+c^2)z^2+(2q'-1)a^2b^2z+a^2b^2c^2.
\]
We claim that $ab$ is the second largest eigenvalue. To prove this, recall that $1$ is the eigenvalue of ${\bm L}$. We can assume $\lambda_1=1$ (because $ab<1$, which means $1$ must be among $\lambda_{1,2,3}$) and hence $f(1)=0$. In addition, $f(z)$ is monic and $f(0)>0$. This means it suffices to show $f(-ab)\leq 0$ and $f(ab)\leq 0$, which are true since 
\begin{align*}
f(ab)&=-a^2b^2q(a-b)^2<0, \qquad
f(-ab)=-a^2b^2q(a+b)^2<0.
\end{align*}

It remains to prove $ab=x^{(d-1)/2}\frac{x-1}{x^d-1}<1/d$ which follows from $\frac{x^d-1}{x-1}=x^{d-1}+\hdots+1>dx^{(d-1)/2}$, where the last inequality is  an application of the AM-GM inequality when $x>1$.
\end{proof}

\subsection{(In)stability of \texorpdfstring{$(q,0,0)$}{(q,0,0)} Fixpoints} \label{sec:q00_stable}

Set $x:=R_0/R_1$ and $y:=C_0/C_1$. Then by rewriting the tree recursion, one can see $x,y$ satisfies the system
\begin{equation} \label{equ:q00_2spin}
\begin{aligned}
x=t^d\left(\frac{ty+q}{ty+q-1}\right)^d,\qquad 
y=t^d\left(\frac{tx+q}{tx+q-1}\right)^d.
\end{aligned}
\end{equation}
Before analysing the stability of the original $(q+1)$-spin system, we first need to study this $2$-spin system. By replacing $\beta:=t/q$, $\gamma=(q-1)/t$ and $\lambda=q^d$, the system above is actually the tree recursion of a general anti-ferromagnetic Ising model with parameter $(\beta,\gamma,\lambda)$. It follows that such system has either one solution $(Q^*,Q^*)$ (uniqueness) or three solutions $(Q^+,Q^-),(Q^*,Q^*),(Q^-,Q^+)$ (non-uniqueness) where $Q^+>Q^*>Q^-$ (see \cite[Section 6.2]{martinelli2007fast} or \cite[Theorem 7]{galanis2016inapproximability}). First and foremost, if $d\geq 5q^k$, the system \eqref{equ:q00_2spin} is actually the latter case. 
\begin{lemma} \label{lem:q00_2spin_non}
 When $q\geq 4, k\geq 2$ and $d\geq 5q^k$, the system \eqref{equ:q00_2spin} lies in non-uniqueness region.
\end{lemma}

One way to prove Lemma \ref{lem:q00_2spin_non} is to verify the non-uniqueness condition in \cite{LLY13}.  However, in our case, that would cause pages of tedious calculation, and we could not get crucial quantitative information about solutions, which is the key to the stability of the original $(q+1)$-spin system. Hence, we show the non-uniqueness by locating the solutions directly, as in the next two lemmas. Also note that, when $x=R_0/R_1=C_0/C_1$, the two-step recursion (\ref{equ:q00_2spin}) can be simplified into the following one-step recursion
\begin{equation} \label{equ:sta_x_sol}
    x=\left(\frac{t^2x+qt}{tx+q-1}\right)^d.
\end{equation}

\begin{lemma} \label{lem:q00_2spin_eq}
Let $(x,x)$ be the solution of (\ref{equ:q00_2spin}) i.e., $x$ be the solution of (\ref{equ:sta_x_sol}). When $q\geq 4, k\geq 2$ and $d\geq 5q^k$, it holds that $tx+q-1<d$. 
\end{lemma}
\begin{lemma} \label{lem:q00_2spin_ineq}
When $q\geq 4, k\geq 2$ and $d\geq 5q^k$, there exists a solution $(x,y)$ to (\ref{equ:q00_2spin}) satisfying (a) $x>y$, and (b) $x>\frac{d}{q^k-q}\cdot d$.
\end{lemma}
We give the proof of Lemmas \ref{lem:q00_2spin_eq} and  \ref{lem:q00_2spin_ineq}  in Section \ref{sec:proof_q00_2spin}.
\begin{proof}[Proof of Lemma \ref{lem:q00_2spin_non}]
This directly follows from Lemma \ref{lem:q00_2spin_eq} and Lemma \ref{lem:q00_2spin_ineq}. 
\end{proof}

 Now we are ready to analyse the stability of $(q,0,0)$-type fixpoints. In the following it will be convenient to let $\Jb$ be the $q\times q$ matrix with 0s on the diagonal and 1s elsewhere, and $\ones$ to be the $q\times 1$ vector with all ones.

\begin{proof}[Proof of Lemma \ref{lem:q00_inequal_stable}]
Let $x=R_0/R_1$ and $y=C_0/C_1$ be the solution of (\ref{equ:q00_2spin}) with $x>y$. Set $a:=\sqrt{\frac{1}{tx+q-1}}$, $b:=\sqrt{\frac{1}{ty+q-1}}$, $r:=\sqrt{\frac{ty}{tx+q}}$ and $s:=\sqrt{\frac{tx}{ty+q}}$. By applying the formula in Lemma \ref{lem:gsv_4_16}, the $(q+1)\times(q+1)$ matrix ${\bm A}$ can be written in the block form
\[
{\bm A}=
\bigg[\begin{array}{cc} rs & as \transpose{\ones}\\ br \ones & ab \Jb\end{array}\bigg].
\] The eigenvalues of ${\bm L}=\left[\begin{smallmatrix} 0&{\bm A}\\{\bm A}^{\top}&0 \end{smallmatrix}\right]$ consist of $\pm ab$ (with multiplicity $q-1$ respectively) and $\pm \lambda_1,\pm \lambda_2$, where $\pm \lambda_1,\pm \lambda_2$ are the zeros of the following biquadratic function
\[
f(z)=z^4-((q-1)^2a^2b^2+qb^2r^2+qa^2s^2+r^2s^2)z^2+a^2b^2r^2s^2.
\]
Again, we assume $\lambda_1=1$ (note that $ab\neq 1$). By Vieta's formula, $\lambda_2=abrs$. Since $rs<1$, this means $ab$ is the second largest eigenvalue. 
Now it suffices to prove $ab<1/d$, which is equivalent to showing $(tx+q-1)(ty+q-1)>d^2$. Note that $ty>t^{d+1}=q^k-q$, and Lemma \ref{lem:q00_2spin_ineq} gives $x>d\frac{d}{q^k-q}$. Therefore $(tx+q-1)(ty+q-1)>txy>d^2$.
\end{proof}

\begin{remark} \label{rmk:q00_vs_2spin}
It is worth noting that the Jacobian stable fixpoints of the system (\ref{equ:q00_2spin}) do not necessarily induce $(q,0,0)$-type Jacobian stable fixpoints of the original $(q+1)$-spin system. This is because the eigenvalue $ab$ from the $(q+1)$-spin system is missing in the $2$-spin system. Interestingly, by directly applying results over $2$-spin system (e.g., \cite[Lemma 8]{galanis2016inapproximability}), what we get is $abrs<1/d$ instead of $ab<1/d$. 
There is an interval of $d$ such that the former holds but the latter does not.
Thus here we cannot only analyze the simplified $2$-spin system.
\end{remark}

\begin{proof}[Proof of Lemma \ref{lem:q00_equal_unstable}]
According to the formula in Lemma \ref{lem:gsv_4_16}, we construct the following $(q+1)\times(q+1)$ matrix ${\bm A}$ with block form 
\[
{\bm A}=
\bigg[\begin{array}{cc} b & \sqrt{ab} \transpose{\ones}\\ \sqrt{ab} \ones & a \Jb\end{array}\bigg]
\]
where $a:=\frac{1}{q-1+tx}$, $b:=\frac{tx}{tx+q}$, and $x$ is the solution of equation (\ref{equ:sta_x_sol}). Because ${\bm A}$ is symmetric, the spectral radius of ${\bm L}=\left[\begin{smallmatrix} 0&{\bm A}\\{\bm A}^{\top}&0 \end{smallmatrix}\right]$ is the same as that of ${\bm A}$. It is not hard to see that $-a$ is an eigenvalue of ${\bm A}$ by multiplicity $q-1$. 
From Lemma~\ref{lem:q00_2spin_eq}, we have that $1/d<a$, and $a<1$ from $q\geq 2$ and $x>0$. Therefore, the fixpoint is unstable. 
\end{proof}

\subsection{\texorpdfstring{$(q,0,0)$}{(q,0,0)} Fixpoint Is Not Maximal} \label{sec:q00_not_max}

Let $q_1=q,q_2=q_3=0$ and $R_0/R_1\neq C_0/C_1$. 
Due to stability, it is difficult to analyse this kind of fixpoint's global optimality 
(recall that it corresponds to a local maxima of $\Psi_1$).
However, observe that changing the value of $R_3$ and $C_3$ will not affect the value of $\overline{\Phi^S}$. 
Therefore, we can force $R_3$ and $C_3$ to be subject to \eqref{equ:phi_deri_zero_R} and \eqref{equ:phi_deri_zero_C}. 
As we will show later, doing so allows us to reuse some lemmas we have utilized in our argument regarding $2$-maximal fixpoints, 
among which the most important one is the perturbation argument. 
We define $r_0,r_1,c_0,c_3$ analogously, and without loss of generality, suppose $r_1,c_3>1$. 

The next proposition shows how we choose $r_1$ and $c_3$. 
\begin{lemma} \label{lem:r1_c3_yields_r0_c0}
Let $x=r_1$ and $y=c_3$ be a pair of solutions to the following system
\begin{equation} \label{equ:r1_c3_sys}
\begin{aligned}
    f_1(x,y):=(x-1)\left(\left(1+\frac{x^d(y-1)}{x^d-1}\right)^dt^{d+1}+q-y^d\right)-y^d+1&=0;\\
    f_2(x,y):=(y-1)\left(\left(1+\frac{y^d(x-1)}{y^d-1}\right)^dt^{d+1}+qx^d-x^d\right)-x^d+1&=0,
\end{aligned}
\end{equation}
with $x,y>1$. Then there exists $r_0$ and $c_0$ such that \eqref{equ:small_r_0_c_0} and \eqref{equ:small_r_1_c_3} are satisfied for $q_1=q,q_2=q_3=0$. 
\end{lemma}

\begin{proof}
The $r_0$ and $c_0$ we choose are defined by 
\begin{align} \label{equ:q00_r0_t}
r_0/t:=\frac{r_1-1}{c_3^d-1}+r_1,\quad  
c_0/t:=\frac{c_3-1}{r_1^d-1}+c_3.
\end{align}
Combining (\ref{equ:q00_r0_t}) with the expression of $f_2(r_1,c_3)=0$, it holds that
\begin{equation*}
c_0^dt+q-c_3^d-\frac{c_3^d-1}{r_1-1}=0, 
\end{equation*}
which is exactly (\ref{equ:r0dt}), and is equivalent to the expression for $r_1$ in \eqref{equ:small_r_1_c_3}. The same argument holds for the $c_3$ expression in \eqref{equ:small_r_1_c_3}. In addition, plugging \eqref{equ:q00_r0_t} back into \eqref{equ:small_r_1_c_3} yields the expressions for $r_0,c_0$ in~\eqref{equ:small_r_0_c_0}.
\end{proof}

Be cautious that we do not assume $R_0/R_1=C_0/C_1$ in \Cref{lem:r1_c3_yields_r0_c0}. Even if we managed to find a pair of solutions $r_1>c_3>1$ to \eqref{equ:r1_c3_sys}, it does not imply that we can find $R_3$ and $C_3$ for the case $R_0/R_1\neq C_0/C_1$, because it is possible for such a pair to correspond to the other case $R_0/R_1=C_0/C_1$. We will handle this in Lemma \ref{lem:sym_c3_bound} after finding a special solution to \eqref{equ:r1_c3_sys}. 

To study the solution of the system \eqref{equ:r1_c3_sys}, 
we need to look into the properties of both functions. 
To clarify the intuition of our approach, 
we plot both functions for the case $q=6,k=3,d=5q^k$ (see Figure \ref{fig:sol}a).
In this setting, the two functions have three intersections in the region $(1,+\infty)^2$: 
one above $y=x$, one near $y=x$ (but still below $y=x$; see Figure \ref{fig:sol}b) and one far below $y=x$. 
Experimentally, only the first two intersections correspond to the case $R_0/R_1\neq C_0/C_1$. 
Hence we would only be interested in them. 
Moreover, as we will see at the end of this subsection,  
a solution such that $x>y$ is required. 
For this purpose, the rest of the subsection endeavours to prove the existence of the intersection near $y=x$ before finishing the proof of Lemma \ref{lem:q00_not_max}. 
Doing so also avoids the need of fully characterising the shape of both curves $f_i(x,y)=0$. 

\begin{figure}[H]
    \centering
    \begin{minipage}{0.45\textwidth}
    \centering
    \begin{tikzpicture}
    \node at (0,0) {\includegraphics[width=0.95\textwidth]{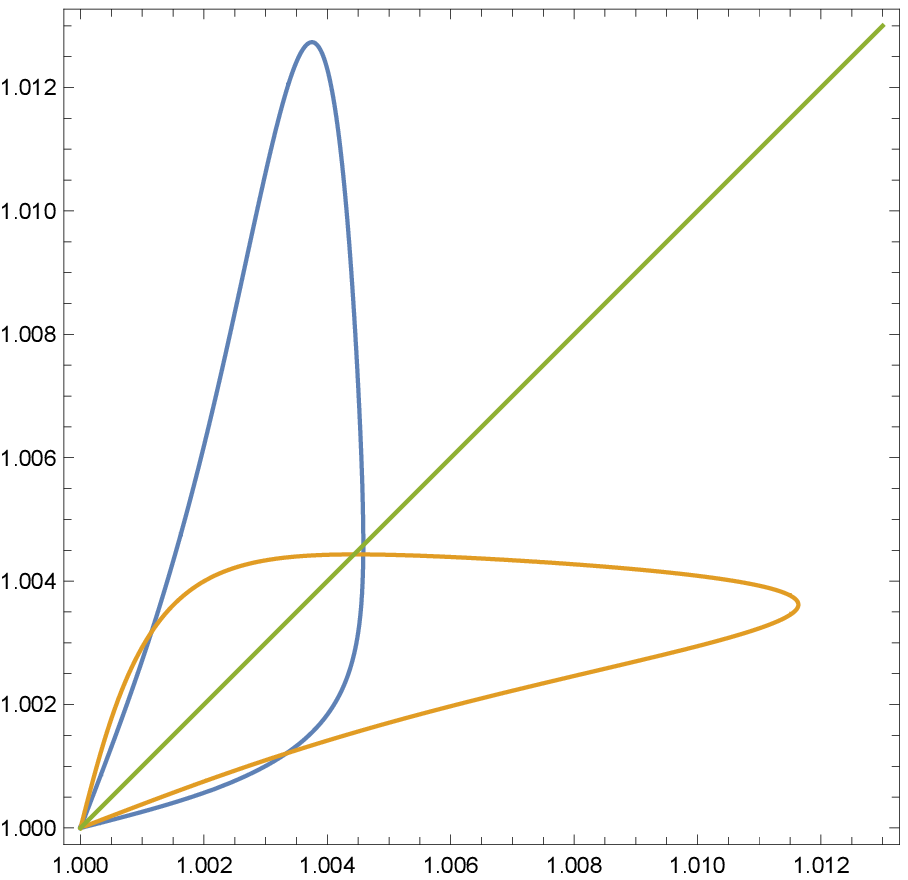}};
    \node at (0.1,2.5) {\footnotesize $f_1(x,y)=0$};
    \node at (2.0,-0.6) {\footnotesize $f_2(x,y)=0$};
    \node at (2.2,1.5) {\footnotesize $y=x$};
    \end{tikzpicture}
    \caption*{\footnotesize (a)}
    \end{minipage}
    \begin{minipage}{0.45\textwidth}
    \centering
    \begin{tikzpicture}
    \node at (0,0) {\includegraphics[width=0.95\textwidth]{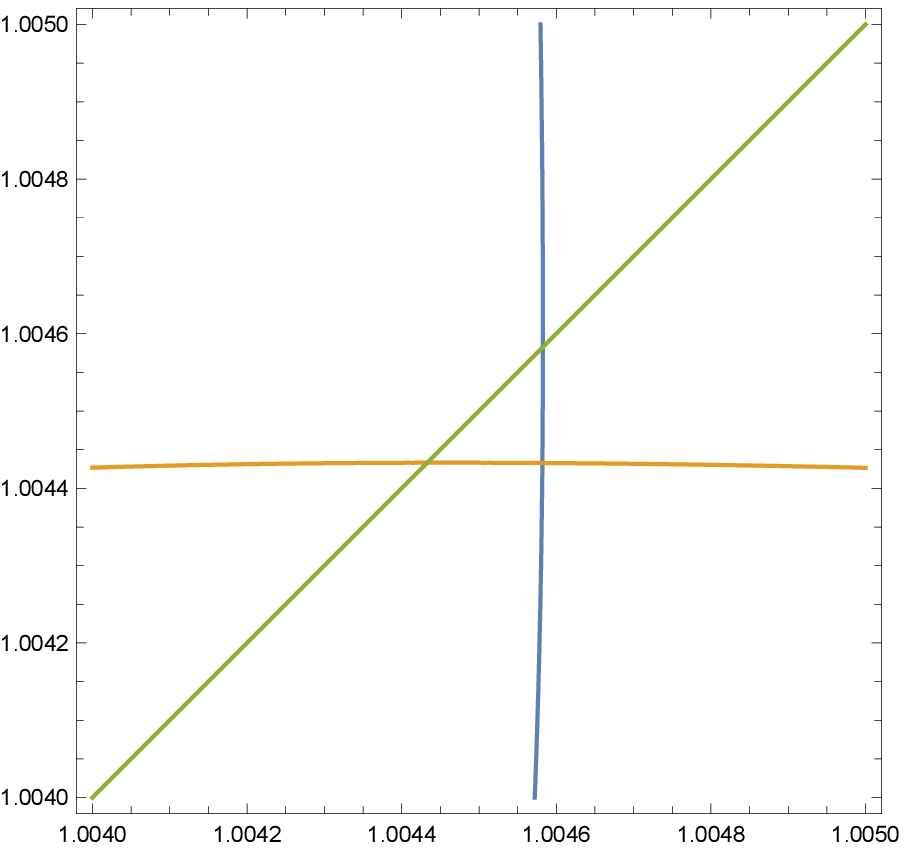}};
    \node at (-0.1,2.5) {\footnotesize $f_1(x,y)=0$};
    \node at (2.3,-0.6) {\footnotesize $f_2(x,y)=0$};
    \node at (2.2,1.3) {\footnotesize $y=x$};
    \end{tikzpicture}
    \caption*{\footnotesize (b)}
    \end{minipage}
    \caption{\footnotesize (a): Shape of the curve $f_1(x,y)=0$, $f_2(x,y)=0$, and $y=x$. (b): Zoom in on the intersection near $y=x$.}
    \label{fig:sol}
\end{figure}

Now we formalize our argument. 
Note that, by mimicking the proof of Lemma \ref{lem:r1_unique}, 
one can show $f_2(x,x)=0$ has exactly one solution $x^{**}>1$. 
Moreover, for any $x\in(1,x^{**})$, $f_2(x,x)<0$, and for any $x>x^{**}$, $f_2(x,x)>0$. A detailed proof is given in Section \ref{sec:proof_unique_root}. 
\begin{lemma} \label{lem:h2_unique}
  For any $q\geq 4, k\geq 2$ and $d\geq 3q^k$, the function
  \[
  h_2(x):=\left(\frac{x^{d+1}-1}{x^d-1}\right)^dt^{d+1}-\frac{x^d-1}{x-1}+(q-1)x^d
  \] 
  has exactly one root $x^{**}$ in the region $x>1$.
\end{lemma}

For $f_1$, we do not need the uniqueness of its intersection with the line $y=x$. 
\begin{lemma} \label{lem:h1_sol}
  For any $q\geq 4, k\geq 2$ and $d\geq 3q^k$, the function
  \[
  h_1(x):=\left(\frac{x^{d+1}-1}{x^d-1}\right)^dt^{d+1}-\frac{x^d-1}{x-1}+q-x^d
  \] 
  has at least one root in the region $x>1$.
  Let $x^*$ be its smallest root. 
  Then $h_1(x)<0$ for $x\in(1,x^*)$.
  Moreover, $x^*>x^{**}$, and consequently $h_2(x^*)>0$. 
\end{lemma}

\begin{proof}
  The first part of the lemma is similar to the proof of Lemma \ref{lem:h2_unique} and Lemma \ref{lem:r1_unique}, by computing $\lim_{x\to 1}h_1(x)<0$ and $\lim_{x\to +\infty}h_1(x)=+\infty$. To prove the second part, note that $h_2(x)>h_1(x)$ for all $x>1$. 
\end{proof}

The next property will be useful later. 
\begin{proposition} \label{prop:f1_f2_boundary}
If $f_1(x,y)=0$, then $x<1+\frac{1}{t^{d+1}-1}$. If $f_2(x,y)=0$, then $y<1+\frac{1}{t^{d+1}-1}$.
\end{proposition}

\begin{proof}
Suppose $x\geq 1+\frac{1}{t^{d+1}-1}$. Then
\begin{align*}
f_1(x,y)&\geq\frac{1}{t^{d+1}-1}\left(\left(1+\frac{x^d(y-1)}{x^d-1}\right)^dt^{d+1}+q-y^d\right)-y^d+1\\
&>\frac{1}{t^{d+1}-1}\left(y^dt^{d+1}+q-y^d\right)-y^d+1=\frac{q}{t^{d+1}-1}+1>0.
\end{align*}

A similar argument holds for $f_2$. 
\end{proof}

Then we study the shape of $f_1$ below the line $y=x$. 
\begin{lemma}  \label{lem:f1-unique}
  Let $g(x)\defeq \frac{(x^d-1)^d}{(x^{d+1}-1)^{d-1}(x-1)}$ and assume that $d\ge 3 q^k$. Then
  \begin{itemize}
    \item[(a)] there is a unique $x_0\in(1,\infty)$ such that $g(x_0)=t^{d+1}$;
    \item[(b)] for any $1 < x < x_0$, $\frac{\partial f_1}{\partial y}<0$ for $y\in(1,x]$; and
    \item[(c)] $x_0>x^*$, where $x^*>1$ is the smallest solution to $f_1(x,x)=0$ (see Lemma \ref{lem:h1_sol}).
  \end{itemize}
    Moreover, for any $1 < x < x_0$, $f_1(x,y)$ is decreasing for $y\in(1,x]$.
\end{lemma}
\begin{proof}
  We first show that $g(x)$ is decreasing for $x>1$.
  By direct calculation,
  \begin{align*}
    g'(x) = \frac{(x^d-1)^{d-1}}{x(x-1)^2(x^{d+1}-1)^d}\left( x^d d^2(x-1)^2-x(x^d-1)^2 \right)<0,
  \end{align*}
  where the last inequality has already been shown in the proof of Lemma~\ref{lem:2_max_stable} for $x>1$.   Notice that $\lim_{x\rightarrow 1}g(x) = \frac{d^d}{(d+1)^{d-1}}$ and $\lim_{x\rightarrow\infty}g(x)=1$.
  As $\frac{d^d}{(d+1)^{d-1}}>\frac{d}{e} > q^k > t^{d+1} = q^k-q>2$,
  there is a unique $x_0$ such that $g(x_0)=t^{d+1}$ and for $x\in(1,x_0)$, $g(x)>t^{d+1}$.
  This shows part (a).

  For  part (b), we have $\frac{\partial f_1}{\partial y} = -x d y^{d-1} + \frac{(x-1)d x^d t^{d+1}(x^dy-1)^{d-1}}{(x^d-1)^d}$ and thus, $\frac{\partial f_1}{\partial y}<0$ is equivalent to 
  \begin{align*}
    \left( x^{d+1}-\frac{(x^d-1)^{d/(d-1)}}{(x-1)^{1/(d-1)}t^{\frac{d+1}{d-1}}} \right)y< x.
  \end{align*}
  As the range of $y$ we consider is $(1,x]$, we only need to show that
  \begin{align*}
    x^{d+1}-\frac{(x^d-1)^{d/(d-1)}}{(x-1)^{1/(d-1)}t^{\frac{d+1}{d-1}}} < 1,
  \end{align*}
  which, after rearranging, is equivalent to $g(x)>t^{d+1}$.
  This is guaranteed by part (a) of the lemma.

  To prove the third part, by \Cref{lem:h1_sol}, it suffices to show $h_1(x_0)=f_1(x_0,x_0)>0$. Note that $x_0$ satisfies
  \begin{align*}
    x_0^{d+1}-\frac{(x_0^d-1)^{d/(d-1)}}{(x_0-1)^{1/(d-1)}t^{\frac{d+1}{d-1}}} = 1, \mbox{ or equivalently, }    \Big(\frac{x_0^{d+1}-1}{x_0^d-1}\Big)^{d-1}=\frac{x_0^d-1}{t^{d+1}(x_0-1)}.
  \end{align*}
  By multiplying this with $\frac{x_0^{d+1}-1}{x_0^d-1}$, we have 
   $ \left(\frac{x_0^{d+1}-1}{x_0^d-1}\right)^{d}t^{d+1}=\frac{x_0^{d+1}-1}{x_0-1}$ and plugging into the expression for $f_1(x,x)$ we get $f_1(x_0,x_0)=q(x_0-1)>0$,  yielding part (c). 
\end{proof}

By \Cref{lem:f1-unique} (b) and (c), the partial derivative $\partial f_1/\partial y\neq 0$ at all points $(x,y)$ such that $f_1(x,y)=0$ and $1<y\leq x\leq x^*$.
Applying the implicit function theorem, $f_1$ yields a continuous function between $x$ and $y$ in the region $1<y\leq x\leq x^*$. 
\begin{corollary}
  The set 
  $\mathcal{P}_1^{+}:=(1,1)+\{(x,y):f_1(x,y)=0,x\geq y>1,x\leq x^*\}$ 
  forms a continuous curve from $(1,1)$ to $x^*,x^*$, 
  where $x^*>1$ is the smallest solution to $f_1(x,x)=0$. 
\end{corollary}

Regarding the shape of $f_2$, we have the next lemma. 
\begin{lemma} \label{lem:p2_shape}
  For any $1<y<1+\frac{1}{q-1}$, there are at most two $x>1$ such that $f_2(x,y)=0$. 
  Moreover, if $1<y<x^{**}$, where $x^{**}>1$ is the unique value such that $f_2(x^{**},x^{**})=0$ (see Lemma \ref{lem:h2_unique}),
  then there is exactly one $x>y$ such that $f_2(x,y)=0$.
\end{lemma}

\begin{proof}
  The crucial idea of this proof is to study the sign of $f_2(x,y)$ at its critical points w.r.t. $x$ (i.e., $x'$ such that $\partial f_2(x,y)/\partial x=0$ at $x=x'$). 

  Fix $y$ in the range and define $g(x)\defeq f_2(x,y)$. By direct calculation, if $g'(x)=0$, $x'$ satisfies
  \[
    t^{d+1}\left(1+\frac{y^d(x'-1)}{y^d-1}\right)^d=\frac{x'^{d-1}(y-q(y'-1))(x'y^d-1)}{(y-1)y^d}.
  \]
  Plugging it back to $g$, we get
  \[
    g(x')=1-\frac{x'^{d-1}(y-q(y-1))}{y^{d}}. 
  \]
  Because $y-q(y-1)>0$, for any critical point $x'$ of $g$, 
  \begin{itemize}
    \item[(a)] if $x'<\chi$, then $g(x')>0$;
    \item[(b)] if $x'=\chi$, then $g(x')=0$;
    \item[(c)] if $x'>\chi$, then $g(x')<0$, 
  \end{itemize}
  where $\chi$ is defined by $\chi:=\big(\frac{y^d}{y-q(y-1)}\big)^{1/(d-1)}$.

  As $g(x)=f_2(x,y)$ for the fixed $y$, $g(x)$ is a polynomial in $x$. 
  Moreover, $g(1)=(y-1)(t^{d+1}+q-1)>0$ and $\lim_{x\to +\infty}g(x)=+\infty$. 
  It implies that $g(x)$ must have an even number of roots.
  If $g(x)$ does not have any root greater than $1$ then we are done. 
  Otherwise, let $x_1>1$ be the smallest root and $x_2$ be the largest root. 
  \begin{itemize}
    \item If $x_1<\chi$, this means the next critical point $x'\geq x_1$ cannot be $x_1$, or otherwise, $g(x')=0$, contradicting with item (a) above. Therefore, $g(x')<0$, which means $x'>\chi$. If there exists another zero $x'<x_3<x_2$, then either $g'(x_3)<0$ or $g'(x_3)>0$ (otherwise, it contradicts with item (b)). In the former case, there must exist another critical point $x''$ such that $\chi<x'<x''<x_3$ and $g(x'')>0$, which contradicts to item (c). In the latter case, there must exist another critical point $x'''$ such that $x_3<x'''<x_2$ and $g(x''')>0$, violating item (c) as well. 
    \item If $x_1>\chi$, this means all the critical points $x'$ in $[x_1,x_2]$ must have function value $g(x')<0$, which implies there is not any other root in $(x_1,x_2)$. 
    \item If $x_1=\chi$ and $x_1$ is not a critical point, then $g'(x_1)<0$ and the argument of the previous case still applies. 
    \item If $x_1=\chi$ and $x_1$ is a critical point, then for any other critical point (if exists) $x'>x_1$, it must holds that $g(x')<0$. 
      Namely once $g(x)$ becomes positive as $x$ increases, the sign of $g'(x)$ will not change.
      It implies that $x_2$ is the only root larger than $x_1$ in this case. 
      If no critical point $x'>x_1$ exists, then $x_1=x_2$ is the only root. 
  \end{itemize}
  
  In all cases, $g(x)$ has at most two roots greater than $1$. This finishes the first part of the lemma.

  For the second part,
  notice that if $y<x^{**}$ then $g(y)<0$, and recall $\lim_{x\to \infty}g(x)=\infty$. 
  The number of zeros larger than $y$ must be odd, and by the first part, it must be unique. Proposition \ref{prop:f1_f2_boundary} guarantees that $x^{**}<1+\frac{1}{t^{d+1}-1}<1+\frac{1}{q-1}$.
\end{proof}


We then argue there is a point on $\mathcal{P}_1^+$ (except $(1,1)$) such that $f_2$ takes zero. 
To establish this, we first find a point $E$ with $f_2(x_E,y_E)=0$
such that it lies to the right of $\mathcal{P}_1^+$ (with some extra conditions, and later we will apply Lemma \ref{lem:p2_shape}). 
To simplify the calculation, we only consider the case $d=5q^k$. 
The proof of the next lemma consists some detailed calculations, which we postpone till Section \ref{sec:proof_interior_exterior}. 


\begin{lemma} \label{lem:exterior}
  Suppose $d=5 q^k$. There exists a point $E$ with $f_2(x_E,y_E)=0$ such that it lies to the right of $\mathcal{P}_1^+$. More specifically, (a) $y_E=1+\frac{0.5}{t^{d+1}-1}$; (b) $y_E<x^{**}$; and (c) $x_E>1+\frac{1}{t^{d+1}-1}$. 
\end{lemma}

This yields the following lemma.

\begin{lemma} \label{lem:r1_c3_sol}
 Suppose $d=5q^k$. The system (\ref{equ:r1_c3_sys}) has a solution $(x,y)$ such that $x>y>y_E$. 
\end{lemma}

\begin{proof}
Consider the following point $M$ on $\mathcal{P}_1^{+}$: $y_M=y_E$, and $x_M$ is the largest one such that $(x_M,y_M)\in\mathcal{P}_1^{+}$.\footnote{$x_M$ is well-defined. This follows from the fact that $(x^d-1)^df_1(x,y_M)$ is a non-zero polynomial with respect to $x$, thus having a finite number of zeros.} Lemma \ref{lem:exterior} (b) asserts that $y_E<x^{**}$, which allows us to invoke Lemma \ref{lem:p2_shape}: for any $y_E<x<x_E$, we have $f_2(x,y_E)<0$. More specifically, $f_2(x_M,y_M)<0$ because $x_M<x^{*}<1+\frac{1}{t^{d+1}-1}<x_E$, where the second and third inequalities come from Proposition \ref{prop:f1_f2_boundary} and Lemma \ref{lem:exterior} (c) respectively. 

Now consider the path of $\mathcal{P}_1^{+}$ between the point $M$ and $(x^{*},x^{*})$. It is continuous and bounded away from both $x=1$ and $y=1$, and the function $f_2(x,y)$ is continuous over $(1,+\infty)\times(1,+\infty)$. This means as one walks along the path, the value of $f_2$ changes continuously; otherwise, it violates the continuity of $f_2$ by a simple $\varepsilon$-$\delta$ argument. Moreover, by the second part of Lemma \ref{lem:h1_sol}, $f_2(x^{**},x^{**})>0$. This means there must be a point $(x,y)$ on the path such that $f_2(x,y)=0$. Moreover, by the choice of $x_M$, it must hold that $y>y_M=y_E$. 
\end{proof}

Now we argue that the solution we find actually satisfies $R_0/R_1\neq C_0/C_1$.  
\begin{lemma} \label{lem:sym_c3_bound}
If $r_1>c_3>1$ yields $R_0/R_1=C_0/C_1$, then it must hold that $c_3<y_E$. 
\end{lemma}
  
\begin{proof}
In this case, $C_0/C_1$ is the solution to the one-step recursion (\ref{equ:sta_x_sol}). Define $u:=c_0/t=(C_0/C_1)^{1/d}/t$. By rewriting (\ref{equ:sta_x_sol}), one can see that $u$ is the unique solution to the following equation
\[
  h(u):=u-\left(1+\frac{1}{t^{d+1}u^d+q-1}\right)=0. 
\]
Note that $u>1$. Using this notation, (\ref{equ:q00_r0_t}) yields that $  c_3=\frac{u(r_1^d-1)+1}{r_1^d}$, giving that  $c_3<u$. 

It remains to show $u<1+\frac{0.5}{t^{d+1}-1}$. 
Note that the system \eqref{equ:sta_x_sol} has a unique fixpoint, namely that $h(u)$ has a unique solution over $u>1$.
Because $h(1)<0$ and $\lim_{u\to \infty}h(u)=\infty$, it suffices to prove $h(1+\frac{0.5}{t^{d+1}-1})>0$. After plugging in the expression and clearing the denominator, it turns out to be equivalent to
\[
  1+3q-2q^k+(q^k-q)\left(1+\frac{1}{2q^k-2(q+1)}\right)^{5q^k}>0
\]
which is true for any $q\geq 4$ and $k\geq 2$. 
\end{proof}

We can finally conclude Lemma \ref{lem:q00_not_max}. 
\begin{proof}[Proof of Lemma \ref{lem:q00_not_max}]
  Lemma \ref{lem:r1_c3_sol} guarantees the existence of $r_1>c_3>y_E$ satisfying \eqref{equ:r1_c3_sys} with $x=r_1$ and $y=c_3$. By Lemma \ref{lem:r1_c3_yields_r0_c0}, given $r_1$ and $c_3$, we can choose $R_0,R_1,R_3,C_0,C_1,C_3$ to satisfy \eqref{equ:phi_deri_zero_R} and \eqref{equ:phi_deri_zero_C}. Lemma \ref{lem:sym_c3_bound} implies that for this choice, $R_0/R_1\neq C_0/C_1$. Moreover, the $2$-spin system regarding $R_0/R_1$ and $C_0/C_1$ lies in non-uniqueness region, and hence the values of $R_0/R_1$ and $C_0/C_1$ are unique up to the swap of $R$ and $C$ (see Section \ref{sec:q00_stable}). 
  
  Because $R_3$ and $C_3$ are subject to (\ref{equ:phi_deri_zero_R}) and (\ref{equ:phi_deri_zero_C}), the first part of the proof in Lemma \ref{lem:dphi_dqi_expr} still holds, even when $q_3=0$ (since we only require (\ref{equ:r_i_d1_d})). Therefore the expression of $\partial\overline{\Phi^S}/\partial q_3$ still applies. Based on this, by going through the proof of Lemma \ref{lem:3_max_maximizer} (a), we can see (\ref{equ:dq1_dq3}) still holds, i.e.,
  \[
    \sgn\left(\frac{\partial\overline{\Phi^S}}{\partial q_1}-\frac{\partial\overline{\Phi^S}}{\partial q_3}\right)=-\sgn(r_1-c_3). 
  \]
   Hence under this choice, $\partial\overline{\Phi^S}/\partial q_1-\partial\overline{\Phi^S}/\partial q_3<0$. Now consider a new $\mathbf{q}$ vector $(q-\varepsilon,0,\varepsilon)$. When $\varepsilon$ is small enough, the value of $\overline{\Phi^S}$ increases, and feasibility in (\ref{equ:phi_cond}) still holds. Because the value of $\overline{\Phi^S}$ at $\mathbf{q}=(q,0,0)$ is irrelavent to $R_3,C_3$, and the value of $\overline{\Phi^S}$ is the same for all fixpoints of type $(q,0,0)$ and $R_0/R_1\neq C_0/C_1$, it means $\overline{\Phi}$ does not take the maximum at fixpoints of such type. 
\end{proof}

\begin{remark}
  Our approach in fact jumps out of the local area around the fixpoint. Intuitively, the argument considers a new ``imaginary'' fixpoint where $\varepsilon$ portion of the $q$ entries $R_1$ (resp. $C_1$) is changed into $R_3$ (resp. $C_3$, recall that $R_3$ and $C_3$ are bounded away from $R_1$ and $C_1$), and compares its value of the original induced matrix norm with the one of $(R_0,R_1,\cdots,R_1,C_0,C_1,\cdots,C_1)$. This is another reason why optimizing $\overline{\Phi^S}$ over all nonnegative $\mathbf{q}$'s instead of integer $\mathbf{q}$'s helps a lot. 
\end{remark}

\section{Remaining Proofs} \label{sec:fixpoint_remaining_proofs}

\subsection{Proof of Lemma \ref{lem:Potts}}\label{sec:Potts}

We will consider the following computational problem.
Given a graph $G=(V,E)$, for a $q$-colouring $\sigma:V\rightarrow\{1,\dots,q\}$,
let $\Mono(G,\sigma)$ be the number of monochromatic edges under $\sigma$.

\prob{\MaxqCut{}}{A undirected graph $G=(V,E)$}{$\max_{\sigma:V\rightarrow\{1,\dots,q\}}\{\abs{E}-\Mono(G,\sigma)\}$}

Let \MaxCut{} be the $q=2$ version of \MaxqCut{}.
Alimonti and Kann \cite{AK00} showed the following. 
\begin{proposition}  \label{prop:max-cut}
  There is a constant $\delta_0>0$ such that,
  there is no randomized polynomial-time approximation algorithm for \MaxCut{} in cubic graphs with relative error $\delta_0$ unless $\NP=\RP$.
\end{proposition}
Furthermore, Kann, Khanna, Lagergren, and Panconesi \cite{KKLP97} showed the following reduction.

\begin{proposition}  \label{prop:max-q-cut}
  For any $0\le\delta\le 1$,
  if \MaxqCut{} in $\left(\frac{\Delta(q+1)}{2}+\frac{q-1}{2}\right)$-regular graphs can be approximated within relative error $\frac{\delta}{2(q+1)}$ in polynomial-time,
  then \MaxCut{} can be approximated within $\delta$ in polynomial-time for $\Delta$-regular graphs.
\end{proposition}

The original reduction in \cite[Theorem 1]{KKLP97} works for only even $q$ and gives relative error lower bound $\frac{\delta}{2(q-1)}$ instead.
For odd $q$ they used a different reduction to achieve the same lower bound but it does not keep the degrees bounded.
Here we briefly describe how to modify the reduction in \cite[Theorem 1]{KKLP97} such that it works for odd $q$ as well, albeit with a slightly worse relative error lower bound $\frac{\delta}{2(q+1)}$.
For odd $q$, given an instance $G=(V,E)$ for \MaxCut{},
we replace each vertex $v\in V$ by a clique $C_v$ of size $\frac{q+1}{2}$ (instead of $\frac{q}{2}$ in the original reduction),
and replace each edge $(u,v)\in E$ by a bipartite complete graph between $C_v$ and $C_u$.
Moreover, give weight $\frac{q+1}{q-1}d_G(v)$ for edges inside $C_v$ (instead of $d_G(v)$) and keep weight $1$ for all other edges.
It can be verified straightforwardly that the proof still works, except that the parameters $\alpha$ and $\beta$ changed from $(\frac{q(q-1)}{2},\frac{2}{q})$ to $(\frac{(q+1)^2}{2},\frac{2}{q+1})$,
which leads to the worse lower bound $\frac{\delta}{2(q+1)}$.
Finally, Crescenzi, Silvestri, and Trevisan \cite{CST01} showed that for a general class of combinatorial optimization problem, 
including \MaxqCut{}, the weighted and unweighted versions have the same approximation complexity.

\begin{lemma}
  There is a constant $0<C_0<1$ such that for any $q\ge 2$, 
  there is no FPRAS for the $q$-state Potts model with weights $B<q^{-1/C_0}$ in $(2q+1)$-regular graphs unless $\NP=\RP$.
  \label{lem:Potts-bounded}
\end{lemma}
\begin{proof}
  Let $C_0\defeq\frac{5\delta_0}{24}$, where $\delta_0$ is from \Cref{prop:max-cut}.
  We claim that for any $q\ge 2$,
  an FPRAS for $Z_B(G)$ with weight $B<q^{-1/C_0}$ in graphs with degree bound $2q+1$ implies 
  an efficient approximation of \MaxqCut{} within relative error $\eps_0\defeq\frac{\delta_0}{2(q+1)}$ in graphs with the same degree bound.
  Then \Cref{prop:max-cut} and \Cref{prop:max-q-cut} (with $\Delta=3$) imply the lemma.

  Given an instance $G=(V,E)$ to \MaxqCut{},
  assume the maximum value of $q$-cut is $\Opt$.
  Let $n\defeq\abs{V}$ and $m\defeq\abs{E}$.
  Then $m=\frac{(2q+1)n}{2}$.
  If we had an FPRAS for the $q$-state Potts model,
  then we can efficiently sample a colouring proportional to its weight.
  (In the local lemma setting, one such reduction is given in \cite{JPV20}.)
  The probability that the cut value of the colouring is less than $(1-\eps_0)\Opt$ is at most
  \begin{align*}
    \frac{q^nB^{m-(1-\eps_0)\Opt}}{B^{m-\Opt} + q^n B^{m-(1-\eps_0)\Opt}}.
  \end{align*}
  In particular, this probability is at most $1/2$ if 
  \begin{align*}
    B^{m-\Opt} \ge q^n B^{m-(1-\eps_0)\Opt},
  \end{align*}
  which is equivalent to $B^{-\eps_0\Opt}\ge q^n$.
  On the other hand,
  notice that a uniformly at random colouring achieves cut value $(1-\frac{1}{q})m$ in expectation.
  Thus, $\Opt\ge(1-\frac{1}{q})m=\frac{(2q+1)(q-1)n}{2q}$.
  Consequently, for any $q\ge 2$, $B^{-\eps_0\Opt}\ge B^{-\delta_0 n \frac{(2q+1)(q-1)}{4q(q+1)}}\ge B^{-C_0n}$,
  since $C_0=\frac{5\delta_0}{24}\le \frac{(2q+1)(q-1)\delta_0}{4q(q+1)}$ for $q\ge 2$.
  Thus if $B<q^{-1/C_0}$, $B^{-\eps_0\Opt}\ge q^n$ as desired.
  Standard methods can boost the success probability from $1/2$ to arbitrarily close to $1$.
\end{proof}

Now we are ready to show \Cref{lem:Potts}.

\begin{proof}[Proof of \Cref{lem:Potts}]
  Given a $(2q+1)$-regular graph $G=(V,E)$, we replace each edge by $s\defeq\floor{\frac{\Delta}{2q+1}}$ parallel edges to get a new graph $G'$ whose degree is at most $(2q+1)s\le \Delta$.
  As $q\ge 2$, $C_1\ge 5$, and $\Delta\ge 2C_1q\ln q$, $s\ge \frac{\Delta}{2q+1} - 1 > 0.63\frac{\Delta}{2q+1}$.

  If we have a Potts model with edge weight $B$ on $G'$,
  then effectively, this is a Potts model on $G$ with $B'=B^s$.
  Thus if $B < 1 - \frac{C_1 q\ln q}{\Delta}$ for $C_1 = 5/C_0$,
  where $C_0$ is from \Cref{lem:Potts-bounded},
  then
  \begin{align*}
    B^{-sC_0} & > \left(1 + \frac{C_1 q\ln q}{\Delta}\right)^{sC_0} \ge e^{\frac{0.8 s C_0C_1 q\ln q}{\Delta}} \ge  e^{\frac{0.8*0.63 C_0C_1 q\ln q}{2q+1}} > e^{\frac{C_0C_1 q\ln q}{2(2q+1)}} \ge q^{0.2C_0C_1} \ge q,
  \end{align*}
  where in the first line we used $1+x\ge e^{0.8 x}$ for $x\le 0.5$.
  Thus this parallel construction can reduce from the Potts model satisfying the conditions of \Cref{lem:Potts-bounded},
  which is $\NP$-hard to approximate.
\end{proof}

\subsection{Proof of Lemmas \ref{lem:dphi_dqi_expr} and \ref{lem:maximal_to_fixpoint}} \label{sec:proof_d_expr_equations}

\begin{proof}[Proof of Lemmas \ref{lem:dphi_dqi_expr} and \ref{lem:maximal_to_fixpoint}]
We first prove Lemma \ref{lem:dphi_dqi_expr}. Let 
\begin{gather*}
    S:=R_0C_0t^2+\left(\mbox{$\sum_{j=1}^{3}$\,}C_jq_j\right)R_0t+\left(\mbox{$\sum_{j=1}^{3}$\,}R_jq_j\right)C_0t+\left(\mbox{$\sum_{j=1}^{3}$\,}R_jq_j\right)\left(\mbox{$\sum_{j=1}^{3}$\,}C_jq_j\right)-\left(\mbox{$\sum_{j=1}^{3}$\,}R_jC_jq_j\right),\\
    R:=R_0^{(d+1)/d}+\left(\mbox{$\sum_{j=1}^{3}$\,}R_j^{(d+1)/d}q_j\right),\quad C:=C_0^{(d+1)/d}+\mbox{$\sum_{j=1}^{3}$\,}C_j^{(d+1)/d}q_j
\end{gather*}
By direct calculation, 
\begin{align*}
  \frac{\partial\overline{\Phi^S}}{\partial q_i }=\tfrac{(d+1)}{S}\left[R_iC_0t+R_0C_it-R_iC_i+R_i\left(\mbox{$\sum_{j=1}^{3}$\,}C_jq_j\right)+C_i\left(\mbox{$\sum_{j=1}^{3}$\,}R_jq_j\right)\right]-d\big(\tfrac{ R_i^{(d+1)/d}}{R}+\tfrac{ C_i^{(d+1)/d}}{C}\big).
\end{align*}
Note that if $q_i>0$ and $R_i\neq 0$, then it must holds that $\partial\overline{\Phi^S}/\partial R_i=0$, and hence (\ref{equ:phi_deri_zero_R}) applies, which gives 
\begin{equation} \label{equ:r_i_d1_d}
\begin{aligned}
  R_0^{(d+1)/d}&\propto R_0(C_0t^2+(q_1C_1+q_2C_2+q_3C_3)t),\quad 
  R_i^{(d+1)/d}&\propto R_i(C_0t+q_1C_1+q_2C_2+q_3C_3-C_i).
\end{aligned}
\end{equation}
Therefore, $\frac{R_i^{(d+1)/d}}{R}=\frac{R_iC_0t+R_i\left(\sum_{j=1}^{3}C_jq_j\right)-R_iC_i}{S}$ and, similarly, $\frac{C_i^{(d+1)/d}}{C}=\frac{C_iR_0t+C_i\left(\sum_{j=1}^{3}R_jq_j\right)-R_iC_i}{S}$. Note that these two equations also hold trivially when $R_i=0$ or $C_i=0$, respectively. Putting these together yields the desired expression for $\frac{\partial\overline{\Phi^S}}{\partial q_i}$ in Lemma~\ref{lem:dphi_dqi_expr}.  

For the second part of Lemma~\ref{lem:dphi_dqi_expr}, without loss of generality, suppose $q_1,q_2>0$ and $\partial\overline{\Phi^S}/\partial q_1-\partial\overline{\Phi^S}/\partial q_2>0$. Take a positive $\varepsilon$ and consider $(q_1+\varepsilon,q_2-\varepsilon,q_3)$. When $\varepsilon$ is small enough, the entries $q_1+\varepsilon$ and $q_2-\varepsilon$ are positive, the value of $\overline{\Phi^S}$ increases, and feasibility in (\ref{equ:phi_cond}) still holds. Hence $(q_1,q_2,q_3)$ does not maximize $\overline{\Phi}$.

Finally we prove Lemma \ref{lem:maximal_to_fixpoint}. Here we have an extra condition that $\mathbf{q}$ is $m$-maximal. This means there exists a maximizer $\mathbf{r},\mathbf{c}$ such that for every $i\neq j$ such that $q_i,q_j>0$, it holds that $R_i\neq R_j$ and $C_i\neq C_j$. From  \eqref{equ:phi_deri_zero_R} and \eqref{equ:phi_deri_zero_C}, we obtain that  $\mathbf{r},\mathbf{c}$ specify an $m$-supported fixpoint of the tree recursion~\eqref{equ:recursion}. 
\end{proof}

\subsection{Proof of Lemma \ref{lem:bad_triple}} \label{sec:bad_triple}

\begin{proof}[Proof of Lemma \ref{lem:bad_triple}] We first show that the maximum in (\ref{equ:phi_q_def}) cannot be achieved at $R_0=0$ or $C_0=0$. Assume otherwise. If $R_0=0$, we have that 
\[
\frac{\partial\overline{\Phi^S}}{\partial R_0}\bigg|_{R_0=0}=\frac{(d+1)t}{S}\cdot(C_0t+q_1C_1+q_2C_2+q_3C_3)>0
\]
where $S>0$. 
Therefore, increasing $R_0$ by a sufficiently small amount increases also the value of $\overline{\Phi^S}$, contradiction. An analogous argument applies for $C_0$. 

Next, we show that at least one of $R_1,R_2,R_3,C_1,C_2,C_3$ are non-zero. Assume otherwise, then

\[
  \frac{\partial\overline{\Phi^S}}{\partial R_1}\bigg|_{R_1=0}=\frac{d+1}{S}\cdot(C_0t+(q_1-1)C_1+q_2C_2+q_3C_3)=\frac{d+1}{S}C_0t>0, 
\]
and therefore we obtain a contradiction as above.

Consider now a triple $(q_1,q_2,q_3)$ with positive entries, and assume w.l.o.g.~that the maximum is taken when $R_1=0$.  We claim that $C_1>0$. Otherwise, by the first part of Lemma \ref{lem:dphi_dqi_expr}, 
we have $\partial\overline{\Phi^S}/\partial q_1=0$, 
and $\partial\overline{\Phi^S}/\partial q_i>0$ for some $i\in\{2,3\}$ since we cannot have $R_2=R_3=C_2=C_3=0$. This yields a contradiction to the second part of Lemma \ref{lem:dphi_dqi_expr}, and therefore $C_1>0$. Observe also that
\[
\frac{\partial\overline{\Phi^S}}{\partial R_1}\bigg|_{R_1=0}=\frac{d+1}{S}\cdot(C_0t+(q_1-1)C_1+q_2C_2+q_3C_3),
\]
so by the argument above we conclude that $C_0t+(q_1-1)C_1+q_2C_2+q_3C_3\leq 0$ and therefore $q_1<1$ (since $C_0,C_1>0$). This yields that
\[
C_1\geq\frac{1}{1-q_1}(C_0t+q_2C_2+q_3C_3)>C_0.
\]
On the other hand, since both of  $C_0,C_1$ are nonzero,  to achieve the maximum, \eqref{equ:phi_deri_zero_C} must hold for $i=1$, which gives $C_0>C_1$, contradiction. Therefore we have $R_1>0$ for triples with positive entries.

Exactly the same argument works for triples of type $(q_1,0,q_3)$ with $q_1,q_3>0$. For the case $(q,0,0)$, note that $q\geq 4>1$, which means the partial derivatives with respect to both $R_1$ and $C_1$ are positive at $R_1=0$ and $C_1=0$ respectively, and hence the maximum cannot be taken at either $R_1=0$ or $C_1=0$. 

To prove the final part of the lemma, suppose that $q_i,q_j>0$. Since $R_i,C_i,R_j,C_j>0$,  we have that \eqref{equ:phi_deri_zero_R} and \eqref{equ:phi_deri_zero_C} apply, which yields that $R_i=R_j$ iff $C_i=C_j$. 
\end{proof}

\subsection{Proof of Lemma \ref{lem:r1_unique} and Lemma \ref{lem:h2_unique}} \label{sec:proof_unique_root}

\begin{proof}[Proof of Lemma \ref{lem:r1_unique}]
We put the expression of $h$ here for convenient reference. 
\begin{equation}\label{equ:r1_unique_expr}
h(x):=\left(\frac{x^{d+1}-1}{x^d-1}\right)^{d}t^{d+1}-\frac{x^d-1}{x-1}+q'+(q'-1)x^d.
\end{equation}
We have that $h$ is continuous over $x\in (1,+\infty)$ and $\lim_{x\to+\infty}h(x)=+\infty$. Using that $t^{d+1}=t^{\Delta}=q^k-q$, we have that 
\begin{align*}
\lim_{x\downarrow 1}h(x)=\big(\tfrac{d+1}{d}\big)^dt^{d+1}-d+q-1<\mathrm{e}q^k-\mathrm{e}q-d+q-1<\mathrm{e}q^k-d<0.
\end{align*}
This implies the existence of  $x$ with $h(x)=0$. To prove the uniqueness of the root, we will show that for any root $x>1$ of $h'(x)$, it holds that $h(x)<0$ (note if such $x$ does not exist then we are already done), using the fact that $h$ is differentiable and its derivative is continuous. To see the reason why it is sufficient, note that the number of roots of $h(x)$ over $x>1$ must be odd (because any critical point of $h$ has value less than zero). Assuming towards contradiction, let $x_2>x_1>1$ be the smallest two roots. Then $h'(x_1)>0$ and $h'(x_2)<0$, indicating there must be some $x^{*}\in (x_1,x_2)$ such that $h'(x^{*})=0$. However, in this case $h(x^{*})>0$, which leads to contradiction. 

Next we prove our claim. Take the derivative of $h$ and let it be zero:
\[
h'(x)=d(q'-1)x^{d-1}-\frac{dx^{d-1}}{x-1}+\frac{x^d-1}{(x-1)^2}+\frac{dt^{d+1}x^{d-1}\left(\frac{x^{d+1}-1}{x^d-1}\right)^{d-1}(d-dx+x(x^d-1))}{(x^d-1)^2}=0,
\]
or equivalently, 
\begin{equation} \label{equ:dh_zero}
\left(\frac{x^{d+1}-1}{x^d-1}\right)^{d}t^{d+1}=\frac{(x^d-1)(x^{d+1}-1)(x-x^d(d(q'(x-1)-x)(x-1)+x))}{d(x-1)^2x^d(d-dx+x(x^d-1))}.
\end{equation}
Combining (\ref{equ:r1_unique_expr}) and (\ref{equ:dh_zero}), we obtain that for any $x$ such that $h'(x)=0$, it holds that \[h(x)=\frac{g(x,d,q')}{d(x-1)^2x^{d-1}(d-dx+x(x^d-1))}\] where
\begin{equation} \label{r1_g_expr}
\begin{aligned}
g(x,d,q')&:=dq'(x-1)^2(x+1)(x^d-1)x^{d-1}-(x^d-1)^2(x^{d+1}-1)\\
&\qquad\qquad\qquad\qquad\qquad-d^2x^{d-1}(x-1)^2(1-x^{1+d}+q'(x-1)(x^d+1)).
\end{aligned}
\end{equation}
It is not hard to see that  $d-dx+x(x^d-1)>0$ for any $x>1$, so,  to show $h(x)<0$, it suffices to prove $g(x,d,q')<0$ for all $x>1$. This will follow by showing that 
\begin{equation}\label{eq:ge4455g}
\mbox{$g(x,d,0)<0$ and $g(x,d,q')$ is decreasing in $q'$, for any $x>1$ and $d\geq 3$, } 
\end{equation}
We have $g(x,d,0)/(x^{d+1}-1)=\big(d^2 (x-1)^2 x^{d-1}-(x^d-1)^2\big)$; the last quantity has been shown negative for all $x>1$ in the proof of Lemma~\ref{lem:2_max_stable}. To prove the monotonicity w.r.t. $q'$ note that
\[
\frac{\partial g}{\partial q'}=-d (x-1)^2 x^{d-1} \left(-(x+1) x^d+d (x-1) \left(x^d+1\right)+x+1\right)=:dx^{d-1}(x-1)^2g_1(x)
\]
where $g_1(x):=-\left(-(x+1) x^d+d (x-1) \left(x^d+1\right)+x+1\right)$. Note that
\[g'_1(x)=(d+1)(x^{d-1}(d+x-dx)-1)<0 \mbox{ for } x>1\] 
Since $g_1(1)=0$, we obtain $g_1(x)<0$ for all $x>1$, proving \eqref{eq:ge4455g} and concluding the proof of Lemma~\ref{lem:r1_unique}. 
\end{proof}

\begin{proof}[Proof of Lemma \ref{lem:h2_unique}]
Recall that $  h_2(x):=\left(\frac{x^{d+1}-1}{x^d-1}\right)^dt^{d+1}-\frac{x^d-1}{x-1}+(q-1)x^d
$. We adopt the same idea as the proof of Lemma \ref{lem:r1_unique} by showing that $h_2$ takes negative values at critical points. The estimation of $\lim_{x\to 1}h_2(x)$ is the same as we did in Lemma \ref{lem:r1_unique}. 

Taking the derivative of $h_2$ and setting it to zero, we get
\[
d q x^{d-1}+\frac{d t^{d+1} \left(\frac{x-1}{x^d-1}+x\right)^{d-1} \left(x \left(x^d-1\right)+d (-x)+d\right) x^{d-1}}{\left(x^d-1\right)^2}-\frac{(d+1) x^d}{x-1}+\frac{x^{d+1}-1}{(x-1)^2}=0,
\]
or equivalently,
\[
\left(\frac{x^{d+1}-1}{x^d-1}\right)^dt^{d+1}=\frac{x^{-d} \left(x^d-1\right) \left(x^{d+1}-1\right) \left(d q x^d-2 d q x^{d+1}+d q x^{d+2}+d x^{d+1}+x^{d+1}-d x^{d+2}-x\right)}{d (x-1)^2 \left(-x^{d+1}+d x-d+x\right)}.
\]
By plugging this back into the expression for $h_2(x)$ and simplifying, we obtain that for any $x$ such that $h_2'(x)=0$ it holds that 
$h_2(x)=\frac{g(x,d,q)}{d (x-1)^2 \left(d-dx+x(x^d-1)\right)}$, where
\[
g(x,d,q):=-d^2 (x-1)^2 \left(x^d (q (x-1)-x)+1\right)+d q (x-1)^2 \left(x^d-1\right)-\left(x^d-1\right)^2 \left(x^{d+1}-1\right) x^{1-d}.
\]
Since $d-dx+x(x^d-1)>0$ for any $x>1$, it remains to prove that $g(x,d,q)<0$. Note that $\frac{\partial g(x,d,q)}{\partial q}=d (x-1)^2 \left(x^d (d (-x)+d+1)-1\right)<0
$ for $x>1$ and therefore $g(x,d,q)<g(x,d,0)$. We also have that $g(x,d,0)/(x^{d+1}-1)=\left(d^2 (x-1)^2-x^{1-d} \left(x^d-1\right)^2\right)<0$, where the inequality follows from the argument below \eqref{eq:ge4455g}. Therefore $g(x,d,0)<0$ for all $x>1$, as desired, finishing the proof. 
\end{proof}

\subsection{Proof of Lemma \ref{lem:q00_2spin_eq} and Lemma \ref{lem:q00_2spin_ineq}} \label{sec:proof_q00_2spin}

We will use the following inequality. 
\begin{equation} \label{equ:exp}
\begin{aligned}
\exp\{a\}>\left(1+\frac{a}{b}\right)^b>\exp\left\{\frac{ab}{a+b}\right\} \qquad \text{for all } a,b>0. 
\end{aligned}
\end{equation}

\begin{proof}[Proof of Lemma \ref{lem:q00_2spin_eq}]
Let $p:=tx+q-1$ and assume for the sake of contradiction that $p\geq d$. Let  $w:=p/q^k$ and $c:=d/q^k$, so tha the assumptions of the lemma imply that $w\geq c\geq 5$.  (\ref{equ:sta_x_sol}) gives
\begin{equation*}
    p=q-1+t^{d+1}\left(1+\frac{1}{p}\right)^dq-1+t^{d+1}\exp\left\{\frac{d}{p}\right\}<q-1+q^k\exp\left\{\frac{c}{w}\right\}. 
\end{equation*}
Therefore, $
w<\frac{q-1}{q^k}+\exp\left\{\frac{c}{w}\right\}<\frac{1}{q^{k-1}}+\mathrm{e}<3$, contradicting $w\geq 5$. 
\end{proof}

\begin{proof}[Proof of Lemma \ref{lem:q00_2spin_ineq}]
For any solution $(x,y)$ of (\ref{equ:q00_2spin}), $x$ satisfies the two-step recursion $f(x)=0$, where
\[
f(z)\defeq t^d\left(1+\frac{1}{t\cdot t^{d}\left(1+\frac{1}{tz+q-1}\right)^d+q-1}\right)^d-z.
\]
Take $x$ as the largest root of $f$. Define $c:=d/q^k$. Because $\lim_{x\to\infty}f(x)=-\infty$, to show (b), it suffices to prove $f\left(c^2q^k\frac{q^k}{q^k-q}\right)>0$, or equivalently, 
\begin{equation} \label{equ:eig_proof_val}
\begin{aligned}
\left(1+\frac{1}{(q^k-q)D+q-1}\right)^d>tc^2\left(\frac{q^k}{q^k-q}\right)^2 \qquad\text{ where }\qquad D:=\left(1+\frac{1}{t\left(c^2q^k\frac{q^k}{q^k-q}\right)+q-1}\right)^d. 
\end{aligned}
\end{equation}
Because $D<\exp\left\{\frac{d}{c^2q^k}\right\}<\exp\{\frac{1}{c}\}<1.2215$, 
\begin{align*}
\text{LHS of (\ref{equ:eig_proof_val})}&>\left(1+\frac{1}{1.2215(q^k-q)+q-1}\right)^{cq^k}>2.2674^c,
\end{align*}
where the last inequality follows from \eqref{equ:exp}. 
Moreover, for any $q\geq 4,k\geq 2,d\geq 5q^k$, we have $(q^k/(q^k-q))^2<1.7778$ and $t<1.0312$. Therefore, RHS of \eqref{equ:eig_proof_val} $<1.8332c^2$, which is smaller than $2.2674^c$ whenever $c\geq 5$. This concludes (b). Part (a) follows from (b) and Lemma \ref{lem:q00_2spin_eq}.
\end{proof}

\subsection{Proof of Lemma \ref{lem:exterior}} \label{sec:proof_interior_exterior}

\begin{proof}[Proof of Lemma \ref{lem:exterior}]
Define $s:=\frac{d}{t^{d+1}-1}$. By Proposition \ref{prop:f1_f2_boundary}, any point on $x=1+\frac{s}{d}$ must be on the right of $\mathcal{P}_1^{+}$. Therefore we are interested in the point $(x,y)$ where $x=1+\frac{s}{d}$ and $y=1+\frac{s}{2d}$. Specifically, we will show $f_2(x,y)<0$, which, together with the fact that $\lim_{x\to +\infty}f_2(x,y)=+\infty$ for any fixed $y>1$, implies the existence of $x_E>x$ such that $f_2(x_E,y)=0$. However, in order to apply Lemma \ref{lem:p2_shape}, we further need to show $y<x^{**}$. The latter can be done by proving $f_2(y,y)<0$ due to Lemma \ref{lem:h2_unique}. 

We deal with the latter one first. Assume $q\geq 4,k\geq 3$, or $q\geq 12,k\geq 2$. Then $5<s<5.4962$ , $\frac{q^k-q}{2(q^k-q-1)}<0.5085$ and $\left(1-\frac{q-1}{2(q^k-q-1)}\right)>0.9580$. Set $D:=\left(1+\frac{s}{2d}\right)^d$. By using \eqref{equ:exp}, one can show $D>\exp\{5/2\}>12.1824$. Therefore, 
\begin{align*}
f_2\left(1+\frac{s}{2d},1+\frac{s}{2d}\right)&=1+\left(1+\frac{s\left(1+\frac{1}{-1+D}\right)}{2d}\right)^d\frac{q^k-q}{2(q^k-q-1)}-D\left(1-\frac{q-1}{2(q^k-q-1)}\right)\\\displaybreak[2]
&<1+\exp\left\{\frac{s}{2}\left(1+\frac{1}{-1+D}\right)\right\}\frac{q^k-q}{2(q^k-q-1)}-D\left(1-\frac{q-1}{2(q^k-q-1)}\right)\\\displaybreak[2]
&<1+0.5085\exp\left\{\frac{5.4962}{2}\left(1+\frac{1}{-1+D}\right)\right\}-0.9580D<0,
\end{align*}
where in the last inequality we use the fact that the function is decreasing in $D$. The cases $(q,k)=(4,2),(6,2),(8,2),(10,2)$ also holds by directly computing $f_2$.

The first one can be handled similarly. Denote $E:=\left(1+\frac{s}{d}\right)^d$. Then $D>E^{1/2}$. By using \eqref{equ:exp} again, $E>\exp\{5\}$. Consider the case $q\geq 8, k\geq 3$, or $q\geq 28, k\geq 2$. Then $5<s<5.1921$, $\frac{q^k-q}{2(q^k-q-1)}<0.5010$ and $\left(1-\frac{q-1}{2(q^k-q-1)}\right)>0.9821$. Therefore, 
\begin{align*}
f_2\left(1+\frac{s}{d},1+\frac{s}{2d}\right)&=1+\left(1+\frac{s\left(1+\frac{1}{-1+D}\right)}{d}\right)^d\frac{q^k-q}{2(q^k-q-1)}-E\left(1-\frac{q-1}{2(q^k-q-1)}\right)\\\displaybreak[2]
&<1+\exp\left\{s\left(1+\frac{1}{-1+D}\right)\right\}\frac{q^k-q}{2(q^k-q-1)}-E\left(1-\frac{q-1}{2(q^k-q-1)}\right)\\\displaybreak[2]
&<1+\exp\left\{s\left(1+\frac{1}{-1+E^{1/2}}\right)\right\}\frac{q^k-q}{2(q^k-q-1)}-E\left(1-\frac{q-1}{2(q^k-q-1)}\right)\\\displaybreak[2]
&<1+0.5010\exp\left\{5.1921\left(1+\frac{1}{-1+E^{1/2}}\right)\right\}-0.9821E<0,
\end{align*}
where in the last inequality we use the fact that the function is decreasing in $E$. The remaining cases $(q,k)=(4,3),(6,3),(4,2),(6,2),\cdots,(26,2)$ also holds by directly computing $f_2$.
\end{proof}

\bibliographystyle{alpha} 
\bibliography{refs.bib}

\end{document}